\documentclass[aps,english,10pt,a4,superscriptaddress,twocolumn]{revtex4}  
\usepackage{amsmath,mathtools}
\usepackage{amsthm}
\usepackage{amssymb}
\usepackage{amsfonts}
\usepackage{bbm}
\usepackage{epsfig}
\usepackage{times}
\usepackage{babel}
\usepackage{color}
\usepackage{hyperref}
\usepackage{mdframed}
\usepackage{changes}
\usepackage{float}
\usepackage{soul}

\def\tr{{\rm tr}}
\def\R{{\mathbb{R}}}
\def\N{{\mathbb{N}}}
\def\H{{\mathcal{H}}}
\def\id{{\mathbf{1}}}

\newcommand{\sunderb}[2]{
  \mathclap{\underbrace{\makebox[#1]{$\cdots$}}_{#2}}
}

\newtheorem{theorem}{Theorem}

\newtheorem{lemma}[theorem]{Lemma}

\newtheorem{corollary}[theorem]{Corollary}

\definecolor{shadecolor}{gray}{0.95}

\begin{document}
\title{Correlating thermal machines and the second law at the nanoscale}
\author{Markus P.\ M\"uller}
\affiliation{Institute for Quantum Optics and Quantum Information, Austrian Academy of Sciences, Boltzmanngasse 3, A-1090 Vienna, Austria}
\affiliation{Departments of Applied Mathematics and Philosophy, University of Western
Ontario, Middlesex College, London, ON N6A 5B7, Canada}
\affiliation{Perimeter Institute for Theoretical Physics, Waterloo, ON N2L 2Y5, Canada}
\date{December 20, 2018}

\begin{abstract}
Thermodynamics at the nanoscale is known to differ significantly from its familiar macroscopic counterpart: the possibility of state transitions is not determined by free energy alone, but by an infinite family of free-energy-like quantities; strong fluctuations (possibly of quantum origin) allow to extract less work reliably than what is expected from computing the free energy difference. However, these known results rely crucially on the assumption that the thermal machine is not only exactly preserved in every cycle, but also kept uncorrelated from the quantum systems on which it acts. Here we lift this restriction: we allow the machine to become correlated with the microscopic systems on which it acts, while still exactly preserving its own state. Surprisingly, we show that this restores the second law in its original form: free energy alone determines the possible state transitions, and the corresponding amount of work can be invested or extracted from single systems exactly and without any fluctuations. At the same time, the work reservoir remains uncorrelated from all other systems and parts of the machine. Thus, microscopic machines can increase their efficiency via clever ``correlation engineering'' in a perfectly cyclic manner, which is achieved by a catalytic system that can sometimes be as small as a single qubit (though some setups require very large catalysts). Our results also solve some open mathematical problems on majorization which may lead to further applications in entanglement theory.
\end{abstract}

\maketitle

\section{Introduction}
Thermodynamics, as it is presented in the textbooks, is usually concerned with macroscopic physical systems, like large ensembles of weakly interacting gas molecules. In this regime, the law of large numbers renders fluctuations mostly irrelevant, and one obtains very precise statistical predictions simply by computing averages. One of the most important quantities in this regime is the Helmholtz free energy,
\[
   F(\rho)=\langle E\rangle_\rho-T\,S(\rho),
\]
where $\langle E\rangle_\rho$ is the average energy of the system in state $\rho$, and $S$ is its entropy. At constant ambient temperature $T$ and constant volume, transitions between two states are possible if and only if the difference between the free energies of the initial and the final state is negative. The free energy difference also tells us how much work we can extract, or need to invest, during a thermodynamic state transition.

However, this formulation of the second law applies only in the thermodynamic limit of large numbers of identically distributed or weakly interacting particles. In contrast, modern technology allows us to probe and manipulate physical systems at much smaller scales~\cite{Faucheux,Toyabe,Baugh,Alemany}, where quantum fluctuations and strong correlations may dominate. Understanding the subtleties of thermodynamics in this regime will also be relevant for some biological processes~\cite{Lloyd,Lambert,Gauger}, since evolutionary pressure tends to force microscopic machines to act as efficiently as possible in thermal environments.

\begin{figure}
\begin{center}
\includegraphics[width=.49\textwidth]{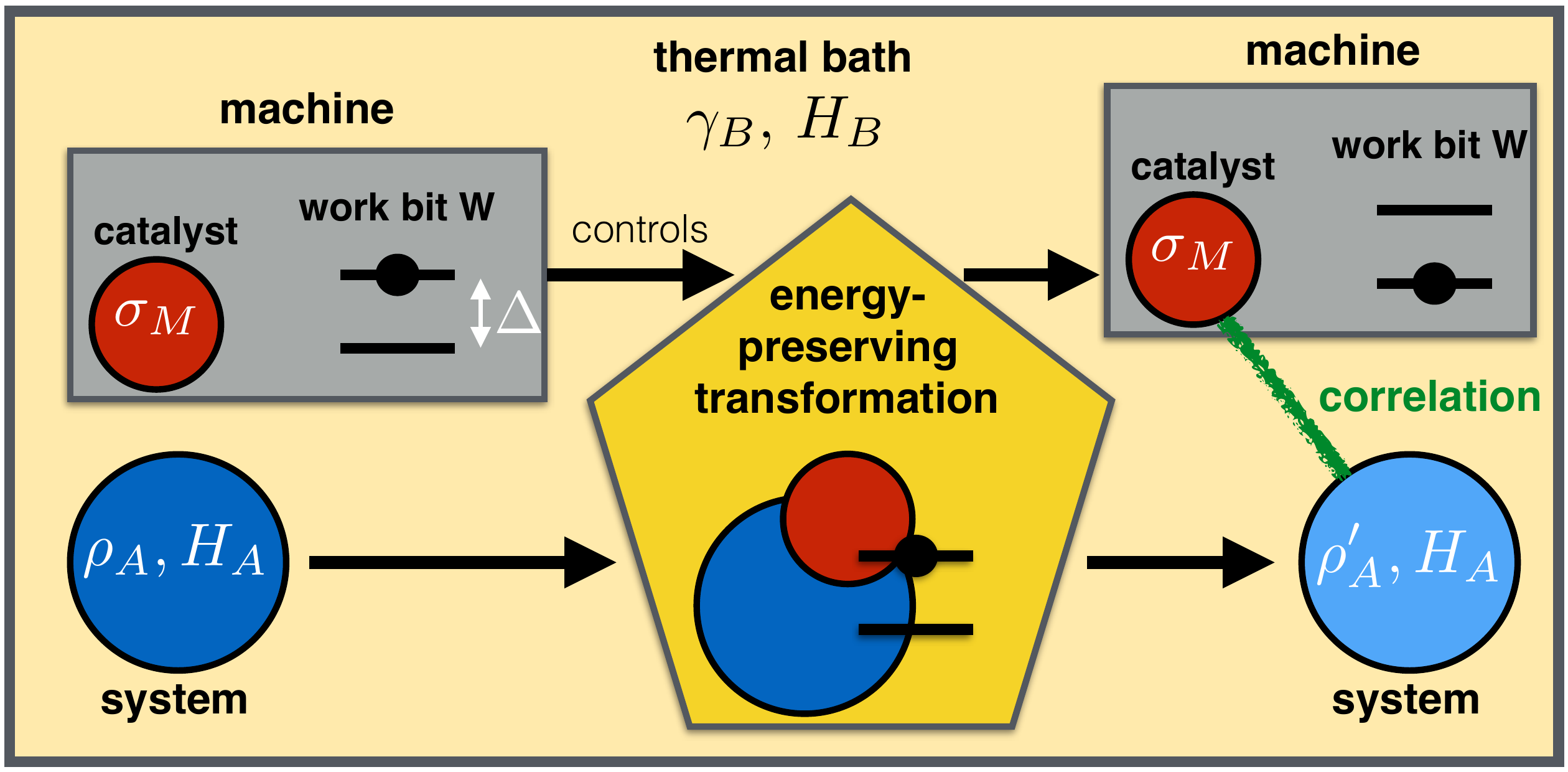}
\end{center}
\caption{\small Thermal operation of the form that we are considering in this paper; compare Figure 1 in~\cite{Brandao}. We have a system $A$ that we would like to act on, by transforming its state $\rho_A$ into another state $\rho'_A$. We have access to a heat bath with an arbitrary Hamiltonian $H_B$, which is in its thermal state $\gamma_B$ at some fixed temperature $T$. The thermal machine contains a quantum system in state $\sigma_M$, and it controls a unitary transformation $U_{AMB}$ (symbolized by the pentagon), acting on the system $A$, heat bath $B$, and its internal system $M$. Crucially, this transformation is fully energy-preserving, i.e.\ $[U_{AMB},H_A+H_M+H_B]=0$. By tracing over the heat bath, we obtain the map $\sigma_{AM}={\rm Tr}_B\left[U_{AMB}\left(\rho_A\otimes\sigma_M\otimes\gamma_B\right)U_{AMB}^\dagger\right]$, which is, in total, a thermal operation, $\rho_A\otimes\sigma_M\mapsto \sigma_{AM}$. We demand that the machine's internal state $\sigma_M$ is exactly preserved (hence $\sigma_M$ is often called a ``catalyst'': it enables the transformation, but is not consumed in the process), and we would like the resulting state of $A$, ${\rm Tr}_M\sigma_{AM}$, to be identical to (or very close to) the desired target state $\rho'_A$. The difference to~\cite{Brandao} is that we allow correlations to build up between $A$ and $M$. If work is spent or extracted, we model this by an additional two-level system (``work bit'') $W$ which, initially as well as finally, is enforced to be exactly in an energy eigenstate (ground state $|g\rangle$ or excited state $|e\rangle$). This ensures that $W$ remains uncorrelated with all other systems that are involved in this process, hence the resulting work $\Delta$ can be reliably transferred to or from an external battery.}\label{fig_sketch}
\end{figure}

With this motivation in mind, based on the techniques and ideas of quantum information theory, an approach to small-scale thermodynamics has recently been developed~\cite{HorodeckiOppenheim,Brandao,Dahlsten,BrandaoSpekkens,Aberg,Faist,Skrzypczyk,Reeb,SSP,Gour,Browne,Masanes,YungerHalpern,Narasimhachar,Frenzel,Renes,Ng,FOR,Egloff,Perry,Alhambra,WilmingGallego} which has demonstrated~\cite{HorodeckiOppenheim,Brandao} that the free energy $F$ looses its role as the unique indicator of state transitions in the microscopic regime. Instead, a family of ``$\alpha$-free energies'' $F_\alpha$ determines the possibility of thermodynamic transformations: a transition is possible if and only if $\Delta F_\alpha\leq 0$ for all $\alpha>0$. In the special case $\alpha=1$, we obtain the standard Helmholtz free energy, $F_1=F$. This recovers the usual second law, $\Delta F\leq 0$, as a special case of an infinite family of ``second laws''. Moreover, the maximal amount of work that can be reliably extracted from a state $\rho$ in contact with a heat bath is given by $F_0(\rho)+k_B T \log Z$, while the minimal amount of work that one has to invest to prepare a state becomes $F_\infty(\rho)+k_B T \log Z$, with $Z$ the partition function and $k_B$ Boltzmann's constant. In general,
\[
   F_0(\rho)<F(\rho)<F_\infty(\rho),
\]
which shows that thermodynamics looses an important reversibility property at the nanoscale: the amount of work needed to create a state exceeds the amount of work that can be extracted from that state. Intuitively, it is the appearance of fluctuations of the order of the free energy itself that is responsible for this effective irreversibility~\cite{Aberg}. It is only in the thermodynamic limit that all $F_\alpha$ become effectively close to $F=F_1$, which recovers standard macroscopic thermodynamics~\cite{BrandaoSpekkens,Brandao,Meer}.

Yet, these recent results all rely on a specific assumption which is, as we will argue, unnecessary in many important physical situations. To understand this assumption, consider transforming a state $\rho_A$ of a physical system $A$ to another state $\rho'_A$ in the presence of a heat bath (see the caption of Figure~\ref{fig_sketch} for more details). This is usually modelled by introducing another system --- a thermal machine $M$, containing a ``catalyst'' $\sigma_M$ --- such that
\begin{equation}
   \rho_A\otimes \sigma_M \mapsto \rho'_A\otimes \sigma_M
   \label{eqCatalytic}
\end{equation}
via some suitable thermal operation. Crucially, the machine starts \emph{and ends} in the same state $\sigma_M$, which means that it is retained in its original form and can be reused, which is essential for a thermodynamic cycle. But we see that, in addition to this crucial property, a further assumption is made: namely, that $A$ and $M$ end up in a product state and \emph{do not become correlated}.

Arguably, there are many situations in which this additional assumption is unwarranted. For example, imagine a microscopic machine that acts on a myriad of small quantum systems, one after the other (say, a stream of particles), and builds up correlations with them while doing so. As long as the machine encounters every system only once, these correlations will not spoil the working of the machine on further systems. This motivates us to consider more general transformations of the form
\begin{equation}
   \rho_A\otimes\sigma_M\mapsto \sigma_{AM},
   \label{eqCatalyticCorr}
\end{equation}
where the reduced final states are $\sigma_A=\rho'_A$ on $A$ and $\sigma_M$ on $M$. That is, the machine's state becomes correlated with the system on which it has acted, but it is locally exactly preserved and can be used again on other systems on which it has not acted before.

Below, we will show that this setting surprisingly restores the standard second law: it is the Helmholtz free energy $F$ that uniquely determines the possible state transitions. In particular, machines that act according to this more general prescription gain a significant advantage: they can essentially tame all fluctuations, and invest or extract the free energy difference with perfect reliability even when operating on single or strongly correlated quantum systems. In some cases, very small catalysts $M$ can already lead to significant improvements of efficiency.

This result answers a major open question of~\cite{Wilming} in the positive: Helmholtz free energy becomes the \emph{``unique criterion for the second law of thermodynamics''}. It is related to the insights of~\cite{Lostaglio}, but goes far beyond them: instead of correlating several auxiliary systems, here the machine becomes correlated with the system on which it acts (but remains otherwise intact), which is arguably a much more natural situation relevant to thermodynamics. The results of this paper also provide new insights into majorization theory, solving several open problems in that field, which may have further applications in entanglement theory~\cite{AMOP}. Namely, majorization determines the possibility of state interconversion for pure bipartite quantum states via local operations and classical communication~\cite{Nielsen}, and standard catalysis is known to enhance the possible transitions~\cite{JonathanPlenio}. Since further thermodynamics-related concepts have recently been translated into this entanglement setting~\cite{AMOP}, we think that the results of the present paper may have interesting implications in this context too. Furthermore, in contrast to earlier results~\cite{BeyondFreeEnergy}, the insights of this paper potentially continue to hold in the presence of quantum coherence (see the conjecture in Subsection~\ref{SubsecCorrExtr}).

\section{Results}
\subsection{Known results without correlation}
\label{SubsecKnown}
We are working within a framework for thermodynamics that is motivated by quantum information theory. This framework formulates thermodynamics as a \emph{resource theory}~\cite{Gour,BrandaoSpekkens}: given any state of a physical system, together with a set of rules that constrain the agent's actions (e.g.\ global energy conservation), a resource theory asks for the ultimate limits of what is possible, e.g.\ how much work the agent can extract or what state transitions she can enforce. A sketch of the setup is given in Figure~\ref{fig_sketch} (for now, ignore the ``work bit'' $W$). We have a collection of quantum systems that each come with their own Hamiltonians. This includes a microscopic system $A$, typically out of equilibrium. We would like to transform its quantum state $\rho_A$ into another state $\rho'_A$, while possibly extracting or investing some work $\Delta\geq 0$. This will be achieved with the help of a thermal machine, as explained in the caption of Figure~\ref{fig_sketch}. Crucially, all processes preserve the total energy exactly (not only its expectation value), and are performed in the presence of a heat bath at fixed temperature $T$. Microscopic reversibility is ensured by modelling global transformations as unitary operations.

As in most previous work (including~\cite{HorodeckiOppenheim} and~\cite{Brandao}), we assume that the decoherence time is much smaller than the thermalization time. This amounts to assuming that all states are block-diagonal in energy (i.e.\ $[\rho_X,H_X]=0$ for all involved quantum systems $X$), which applies to a large variety of situations in physics, including ones traditionally studied in the context of Landauer erasure~\cite{Landauer61,Landauer88}. In this semiclassical regime, the state of any system is characterized by the occupation probabilities of the different energy levels; the state is thermal if these probabilities are given by the Boltzmann distribution. It has recently been shown that coherence significantly complicates the picture~\cite{BeyondFreeEnergy,LostaglioX,Korzekwa,Cwiklinski}; studying the semiclassical regime is therefore a crucial first step even if one is interested in the more general situation with coherence. We thus defer the treatment of quantum coherence to future work, but discuss some evidence that our main result could still hold in the presence of coherence in Subsection~\ref{SubsecCorrExtr}.

In order to account very carefully for all contributions of energy and entropy, we assume that the machine can strictly only perform the following operations: energy-preserving unitaries; accessing thermal states from the bath; and ignoring heat bath degrees of freedom by tracing over them. This results in a class of transformations called \emph{thermal operations} which have the form stated in the caption of Figure~\ref{fig_sketch}. If we assume for the moment that there is no work reservoir $W$, and demand that these operations preserve the local state of the machine $M$ \emph{and also its independence from $A$}, then they describe transitions $\rho_A\to\rho'_A$ as in equation~(\ref{eqCatalytic}). It has recently been shown~\cite{Brandao} that a thermal transformation can achieve this transition (up to an arbitrarily small error on $A$) if and only if all $\alpha$-free energies decrease in the process:
\begin{equation}
   \Delta F_\alpha=F_\alpha(\rho'_A)-F_\alpha(\rho_A)\leq 0 \mbox{ for all }\alpha\geq 0.
   \label{egDeltaFAlpha}
\end{equation}
Here $F_\alpha(\rho_A)=k_B T S_\alpha(\rho_A\|\gamma_A)-k_B T \log Z_A$, with $Z_A$ the partition function of $A$, $T$ the background temperature, $k_B$ Boltzmann's constant, and $S_\alpha$ the R\'enyi divergence~\cite{vanErven} of order $\alpha$ (see Subsection~\ref{SecProof} and Appendix). For $\alpha=1$, this reduces to the well-known Helmholtz free energy $F_1=F$.

\subsection{Example: smaller work cost with a single qubit catalyst}
\label{SubsecExample}
To see that the $\alpha$-free energies impose severe constraints on the workings of a thermal machine, let us look at a simple example. Suppose that a thermal machine is supposed to heat up a system $A$ from its thermal state (of ambient temperature $T$) to infinite temperature. If $A$ is some two-level system with energies $0$ and $E_A$, and the temperature is such that $E_A/(k_B T)=\log 2$, then the initial thermal state is $\gamma_A=\left(\begin{array}{cc}2/3 & 0 \\ 0 & 1/3 \end{array}\right)$. The desired target state is $\rho'_A=\left(\begin{array}{cc} 1/2 & 0 \\ 0 & 1/2 \end{array}\right)$. The associated work cost will be delivered by an additional work bit $W$ with energy gap $\Delta>0$. It starts in its excited state $|e\rangle$ and will end up in its ground state $|g\rangle$. The machine tries to implement the transition
\[
   \gamma_A\otimes |e\rangle\langle e|_W\mapsto \rho'_A\otimes |g\rangle\langle g|_W
\]
with a work cost $\Delta$ that is as small as possible. As before, this is achieved by a catalytic thermal operation of the form
\[
   \gamma_A\otimes\sigma_M\otimes |e\rangle\langle e|_W\mapsto \rho'_A\otimes\sigma_M\otimes |g\rangle\langle g|_W,
\]
cf.~(\ref{eqCatalytic}) and Figure~\ref{fig_example}.

\begin{figure}
\begin{center}
\includegraphics[width=.3\textwidth]{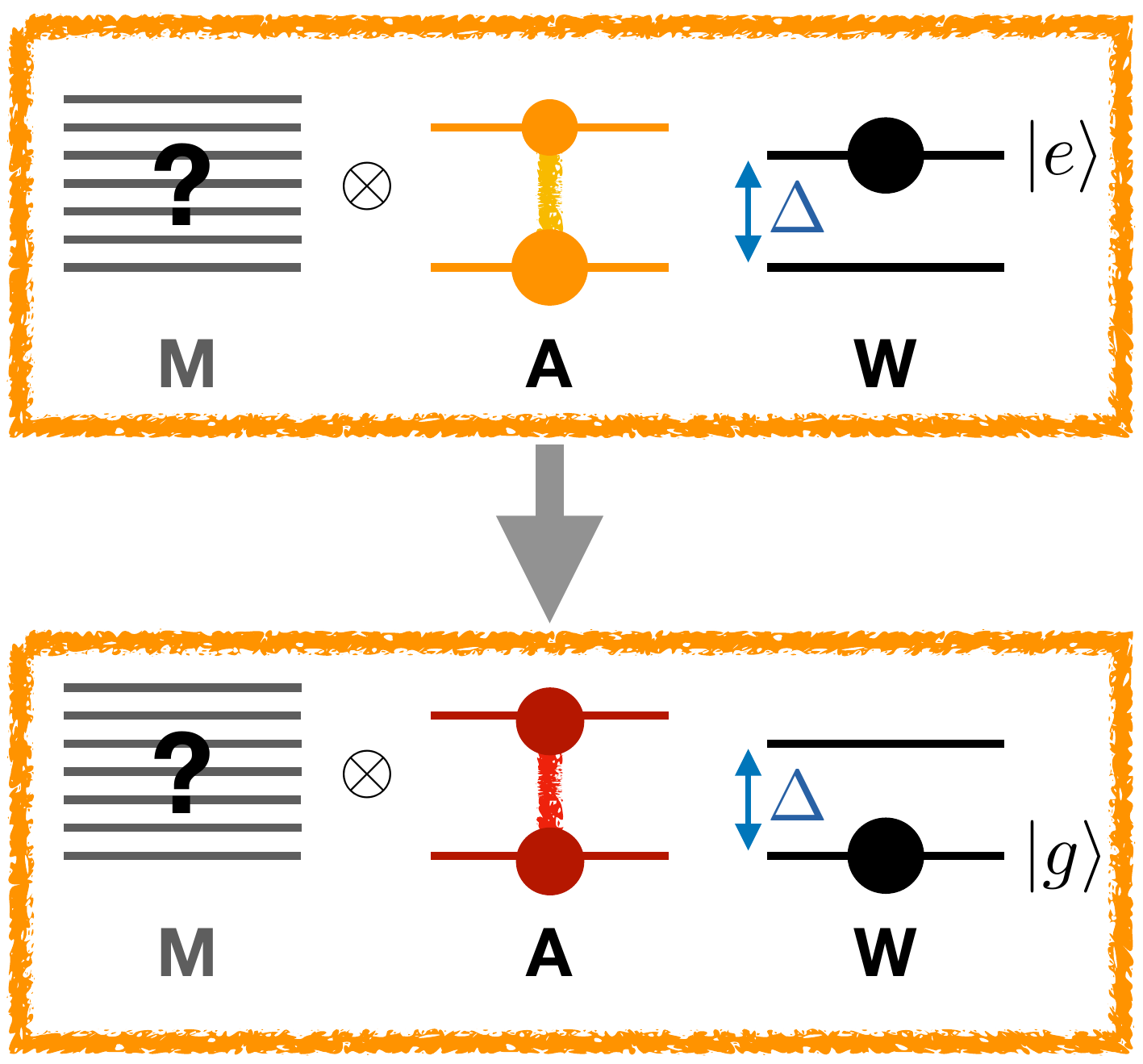}
\end{center}
\caption{\small Example of work cost scenario without allowing correlations to build up. A qubit $A$, initially in equilibrium, is supposed to be heated up to infinite temperature by spending some energy $\Delta$ and by using a (potentially large) catalytic system $M$ that remains uncorrelated with $AW$ (and unchanged by the process). A transition of this form is only possible at work cost of at least $\Delta\gtrapprox .4 k_B T$.}\label{fig_example}
\end{figure}
What is the minimal amount of work needed, i.e.\ the smallest possible $\Delta$? The $\alpha$-free energy difference (see Appendix or~\cite{Brandao} for the definition) between initial and final state of $AW$ turns out to be
\[
   \frac{\Delta F_\alpha}{k_B T} = \frac{\log(2^{1-\alpha}+1)-\alpha\log 2+(\alpha-1)\log 3}{\alpha-1}-\frac\Delta {k_B T}
\]
which is increasing in $\alpha$. Thus, this is $\leq 0$ for all $\alpha$ if and only if $\Delta F_\infty\leq 0$, which becomes
\begin{equation}
   \Delta\geq k_B T \log(3/2)\approx .4\, k_B T.
   \label{eqWorkCost}
\end{equation}
This is the ultimate limit for a transition as shown in Figure~\ref{fig_example} to be successfully implementable. On the other hand, the standard free energy difference is
$\Delta F/(k_B T)=\log 3 -3/2 \log 2 -\Delta/(k_B T)$,
and for this to be $\leq 0$ we must have
\[
   \Delta\geq k_B T\left(\log 3-3/2 \log 2\right)\approx .06\, k_B T.
\]
\begin{figure}
\begin{center}
\includegraphics[width=.3\textwidth]{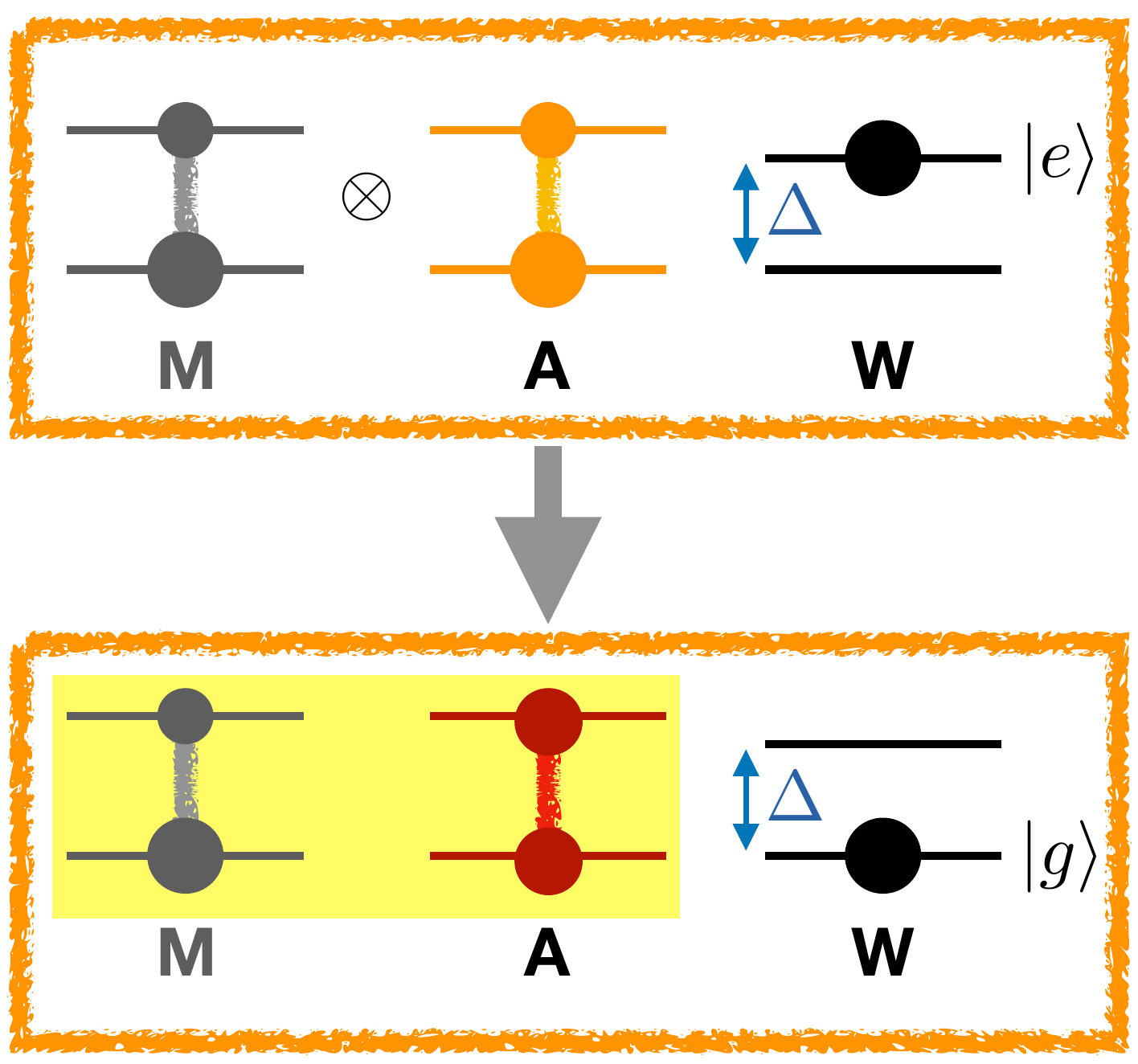}
\end{center}
\caption{\small Example work cost if correlations between $M$ and $A$ \emph{are} allowed to build up. Since $M$ is locally exactly preserved, it can be reused on further states (just not on those ones on which it has already acted before). This transition is possible at work cost of only $\Delta\gtrapprox .26 k_B T$ (about $1/3$ less than in Figure~\ref{fig_example}), even though the catalyst $M$ consists only of a single qubit.}\label{fig_example2}
\end{figure}
Thus, textbook thermodynamic reasoning would suggest that $.06\,k_B T$ of energy should be sufficient for the state transition; however, our analysis has shown that the machine needs to spend considerably more work, namely $.4\, k_B T$. As explained above, one reason for this is that we are dealing with the case of a \emph{single} system only. The standard thermodynamic equations apply to large numbers of (independent, or weakly correlated) identical systems and their averages. That is, if $\Delta^{(n)}$ is the energy (for example, energy gap of the work bit) that is needed to approximately achieve the transition
\[
   \underbrace{\gamma_A\otimes\ldots\otimes\gamma_A}_n\otimes |e\rangle\langle e|_W \mapsto \underbrace{\rho'_A\otimes\ldots\otimes \rho'_A}_n\otimes |g\rangle\langle g|_W,
\]
then $\Delta^{(n)}\approx n\Delta F$ (here $.06\, n k_B T$) as $n$ becomes large (up to corrections that are sublinear in $n$), as shown, for example, in~\cite{BrandaoSpekkens,Janzing}. Intuitively, by acting collectively on a large number of particles, a machine can achieve more than if it had to act on each particle separately. This phenomenon is again related to versions of the law of large numbers, which results in quantities becoming sharply peaked around their averages in large ensembles.

This is bad news for the machine --- what if it is essential for the given physical setup that the specific single instance of $A$ is being heated, and that very little work is spent in this process? A glance at Figure~\ref{fig_example} can guide us towards a solution: whatever transition we have there, it must come from a thermal operation that is being performed globally on the $MAW$ system. While doing so, the thermal machine better takes care of preserving the state of $M$ so that it can be reused in the future. But the way we have formulated catalytic thermal operations so far introduces yet another complication for the working of the machine: it must keep $M$ uncorrelated from $AW$. This seems hard and  overly constraining, given that interaction typically creates correlation.

We thus have two independent motivations to allow correlations between $M$ and $AW$: the difficulty to avoid correlations on interaction, and the desire to achieve higher efficiency. We will now show that the latter goal can indeed be achieved by allowing correlations to build up, \emph{even if the catalyst $M$ is as small as a single qubit}. Suppose that $M$ has a trivial Hamiltonian, $H_M=0$, and two basis states $|0\rangle$ and $|1\rangle$ (both of energy zero). Denote ground and excited state of $A$ by $|g_A\rangle$ and $|e_A\rangle$, and consider the correlated state
\begin{eqnarray*}
   \rho'_{AM}&:=&\frac 1 {10}|g_A 0\rangle\langle g_A 0| + \frac 4 {10} |g_A 1\rangle\langle g_A 1|\\
   	&&+\frac 2 {10} |e_A 0\rangle\langle e_A 0|+\frac 3 {10} |e_A 1\rangle\langle e_A 1|.
\end{eqnarray*}
By computing the partial trace, we find that $\rho'_A$ is indeed the infinite-temperature state, and
\[
   \rho'_M=\frac 3 {10} |0\rangle\langle 0|+\frac 7 {10} |1\rangle\langle 1|=:\sigma_M
\]
which will also be our local qubit catalyst state $\sigma_M$. Thus, if we enforce thermal transitions of the form
\[
   \gamma_A\otimes\sigma_M\otimes |e\rangle\langle e|_W \longrightarrow \rho'_{AM} \otimes |g\rangle\langle g|_W,
\]
then $A$ will be heated up, the local reduced state of $M$ will be preserved, and correlations will build up between $A$ and $M$ (note that there cannot be any correlations with $W$ since it is in a pure state). Now, as we show in Appendix~\ref{AppendixQubit}, this transition can be achieved by a thermal operation (without the need for any additional ``standard'' catalysts), investing only
\[
   \Delta\gtrapprox .26 k_B T
\]
of work. That is, the single qubit catalyst allows us to save about $1/3$ of the total work cost as compared to~(\ref{eqWorkCost}). One can easily imagine situations in which this represents a decisive physical advantage.

In the remainder of the paper, we will explore the ultimate limitations of this kind of ``correlating'' catalysis. We will show that these limitations are uniquely determined by Helmholtz free energy. That is, by using other suitable catalysts in the example above, one can get as close to $\Delta=\Delta F \approx .06 k_B T$ as one wishes (but not below), at the prize of having a possibly large catalyst at hand (which can however be reused).

\subsection{Correlating state transformations in general}
\label{SubsecCorrelating}
Under what conditions can a state transition as in the example above be achieved? For the moment, let us assume that there is no work bit $W$ (we will reintroduce $W$ in the next subsection). In order to implement the transition~(\ref{eqCatalyticCorr}) with a thermal operation, it is still necessary that $\Delta F_\alpha\leq 0$ on the joint system $AM$ for all $\alpha$, since this is a necessary condition for \emph{all} thermal operations. In the uncorrelated case, eq.~(\ref{eqCatalytic}), the same inequalities follow for system $A$ alone, since $F_\alpha(\rho_A\otimes \sigma_M)$ is simply the sum $F_\alpha(\rho_A)+F_\alpha(\sigma_M)$. But in the correlated case, the situation is different. In this case, it turns out that there are two special values of $\alpha$, namely $\alpha=0$ and $\alpha=1$, for which $F_\alpha$ has the important property of \emph{superadditivity}: that is,
\[
   F_\alpha(\sigma_{AM})\geq F_\alpha(\sigma_A)+F_\alpha(\sigma_M),\qquad \alpha=0,1.
\]
This allows us to obtain two conditions on the state of $A$ alone, starting with the non-increase of $F_\alpha$ on $AM$:
\begin{eqnarray*}
0&\geq& F_\alpha(\sigma_{AM})-F_\alpha(\rho_A\otimes \sigma_M)\\
&\geq& F_\alpha(\rho'_A)+F_\alpha(\sigma_M)-F_\alpha(\rho_A)-F_\alpha(\sigma_M).
\end{eqnarray*}
Thus, we conclude that
\[
   F_\alpha(\rho'_A)- F_\alpha(\rho_A)\leq 0,\qquad \alpha=0,1.
\]
But the other $F_\alpha$ are not in general superadditive, as emphasized in~\cite{Lostaglio,GEW,Wilming}, see also~\cite{Aczel,Csiszar}. Hence we cannot draw an analogous conclusion for the other $\alpha$-free energies. Moreover, the condition $F_0(\rho'_A)-F_0(\rho_A)\leq 0$ is arguably physically irrelevant for the purpose of this subsection, as a glance at its definition shows: we have
\[
   F_0(\rho_A)=-k_B T \log Z_A + k_B T S_0(\rho_A\|\gamma_A)
\]
(the ``min-free energy'' from~\cite{HorodeckiOppenheim}), where $S_0(\rho_A\|\gamma_A)=-\log\tr(\pi_{\rho_A} \gamma_A)$ is the ``min-relative entropy'' from quantum information theory~\cite{Datta}, with $\pi_{\rho_A}$ the projector onto the support of $\rho_A$. This is a discontinuous quantity which takes its minimal value whenever the state has full rank, i.e.\ no energy level has probability zero. Since there is no essential physical difference between zero population and extremely small non-zero population, we can ensure that the target state $\rho'_A$ has full rank by allowing an arbitrarily small error in the transition.

Thus, only the standard Helmholtz free energy condition $\Delta F_1\equiv \Delta F \leq 0$ survives as a relevant necessary condition for a correlating state transition. But is it also sufficient --- that is, given that it is satisfied, can we in principle always engineer the machine $M$ and its state such that transition~(\ref{eqCatalyticCorr}) is possible? This was conjectured in Ref.~\cite{Wilming}, and our first main result shows that this is indeed the case:\\

\begin{mdframed}[backgroundcolor=shadecolor,innertopmargin=\topskip]
\textbf{Main Result 1.} \emph{Consider some initial state $\rho_A$ and target state $\rho'_A$, both block-diagonal in energy. In the setting of Figure~\ref{fig_sketch} (without work bit $W$), the transition
\[
   \rho_A\otimes\sigma_M \mapsto \sigma_{AM},
\]
with $\sigma_A:={\rm Tr}_M\sigma_{AM}$ arbitrarily close to $\rho'_A$, can be achieved by a thermal operation if and only if $F(\rho_A)\geq F(\rho'_A)$, with $F$ the Helmholtz free energy. Note that the state $\sigma_M$ of the thermal machine $M$ is exactly identical before and after the transformation, and its state space is finite-dimensional.}

\emph{Furthermore, the Hamiltonian on $M$ can be chosen as $H_M=0$, and the final correlation between $A$ and $M$, as measured by the mutual information $I(A:M)_{\sigma}$, can be made arbitrarily small (but not in general zero).}\\
\end{mdframed}
The proof is sketched in Subsection~\ref{SecProof}, and given in full detail in the Appendix. As in earlier work, the catalyst $\sigma_M$ will in general depend on the initial and final states $\rho_A,\rho'_A$ and on the Hamiltonian $H_A$; it will also depend on the amount of correlation $I(A:M)_{\sigma}$ that the agent is willing to allow to build up. Therefore, we should think of the thermal machine in Figure~\ref{fig_sketch} as containing a \emph{large collection} of different catalysts $\sigma_M$. Depending on the situation, the machine will apply the corresponding suitable catalyst.

Doesn't the agent have to ``know the system state $\rho_A$'' to apply her machine accordingly? The answer to this question is that the state $\rho_A$ is supposed to model the agent's knowledge of the system $A$ in the first place, and this interpretation is chosen implicitly in most works in the present context. For example, the energy cost in Landauer erasure~\cite{Landauer61,Landauer88} is not necessarily relying on an objective ``delocalization'' of a particle in two halves of a box, but is simply due to the agent's missing knowledge about whether it will be detected in the left or the right half in any single experimental run. Consequently, the agent can always choose the catalyst that suits her knowledge of the system as encoded in her state description.

What can we say about the size of the catalyst $\sigma_M$? As we have shown by example in Subsection~\ref{SubsecExample}, in some cases the catalyst can be as small as a qubit and still allow for substantial advantages as compared to the standard ``non-correlating'' notion of catalysis. Main Result 1 formalizes the ultimate possibilities and limitations of thermal machines acting on single small quantum systems, without aiming at the use of ``realistic'' catalysts. Thus, in the proof, we will take advantage of constructing ``custom-tailored'' catalysts that can generically be very large. This is not different, however, from the case of standard catalysis~\cite{Klimesh,Turgut}. We leave the problem of finding efficient implementations of the catalysts for future work.

\subsection{Correlating work cost in general}
\label{SubsecForm}
We now consider the more general situation that we have an additional work reservoir, containing some energy $\Delta\geq 0$ that we may spend in addition to achieve the state transition. As depicted in Figure~\ref{fig_sketch}, this is modelled by a ``work bit'' $W$, a two-level system with energy gap $\Delta$, that will transition from its excited state $|e\rangle$ to its ground state $|g\rangle$ during this process. An example has been discussed in Subsection~\ref{SubsecExample} above.

We imagine that this work bit is part of a larger ``ladder'' of energy levels which we can charge or discharge like a battery in between thermodynamic cycles. It is therefore crucial to demand that the work bit $W$ does not become correlated with the other parts of the machine $M$. One way to ensure this is to demand that $W$ is always exactly, and not just approximately, in an energy eigenstate. It turns out that we can always achieve this behavior:\\

\begin{mdframed}[backgroundcolor=shadecolor,innertopmargin=\topskip]
\textbf{Main Result 2.} \emph{Consider some initial state $\rho_A$ and target state $\rho'_A$, both block-diagonal, such that $F(\rho'_A)\geq F(\rho_A)$. Using a work bit $W$ with some energy gap $\Delta$ larger than, but arbitrarily close to $F(\rho'_A)-F(\rho_A)$, the transition
\[
   \rho_A\otimes\sigma_M\otimes |e\rangle\langle e|_W \mapsto \sigma_{AM}\otimes |g\rangle\langle g|_W
\]
can be achieved by a thermal operation, where $\sigma_A:={\rm Tr}_M\sigma_{AM}$ is arbitrarily close to $\rho'_A$.}

\emph{Similarly as in Main Result 1, the state $\sigma_M$ is exactly identical before and after the transformation, $M$ is finite-dimensional, and the resulting correlations between $A$ and $M$ can be made arbitrarily small.}\\
\end{mdframed}
The method to engineer this transition is very similar to that of Main Result 1, except for one important difference: since we are interested in producing a pure state $|g\rangle$ \emph{exactly}, we have to make sure that the min-free energy $F_0$, which depends only on the rank of the state, is non-increasing in the process. But this holds  automatically because
\[
   F_0(|e\rangle\langle e|_W) > F_0(|g\rangle\langle g|_W)
\]
if $\Delta>0$. Thus, the min-free energy introduces no new constraints in the case that we use work to form a state $\rho'_A$. The ``correlating work cost'' is given by the Helmholtz free energy difference $F(\rho'_A)-F(\rho_A)$.

\subsection{Correlating work extraction, and an open problem}
\label{SubsecCorrExtr}
Consider the converse situation: given an initial state $\rho_A$ and a target state $\rho'_A$ such that $F(\rho_A)\geq F(\rho'_A)$, we would like to extract work by transforming a work bit from its ground state $|g\rangle\langle g|_W$ to its excited state $|e\rangle\langle e|_W$. Here we encounter a problem: since $\rho_A$ will in general have full rank, the work bit alone lower-bounds the min-free energy difference of the corresponding transition, namely $\Delta F_0=F_0(|e\rangle\langle e|_W)-F_0(|g\rangle\langle g|_W)$, and this is a \emph{positive} amount if the energy gap $\Delta$ is positive.

Thus, unfortunately, the min-free energy condition $\Delta F_0\leq 0$ forbids this transition. If we still insist on producing the excited state exactly (for the reasons explained in Subsection~\ref{SubsecForm}), we need an additional resource: namely, a \emph{sink} $S$ for the corresponding entropy $S_0(\rho)=\log{\rm rank}(\rho)$, the ``max entropy''. A max entropy sink $S$ carries a trivial Hamiltonian, $H_S=0$, such that $S_0(\rho_S\|\gamma_S)=\log d_S-S_0(\rho_S)$, where $d_S$ is the Hilbert space dimension of $S$. Thus, we can extract min-free energy by dumping max entropy $S_0$ into $S$, which can be achieved by increasing the rank of the state of $S$. For example, if $S$ carries a state $\tau^{(m,n)}_S$ with eigenvalues $\left(\underbrace{\frac 1 m,\ldots, \frac 1 m}_m,\underbrace{0,\ldots,0}_{n-m}\right)$
and this state is transformed into a state $\tau^{(m,n,\varepsilon)}_S$ (for some small $\varepsilon>0$) with eigenvalues $\left(\underbrace{\frac{1-\varepsilon}m,\ldots,\frac{1-\varepsilon}m}_m,\underbrace{\frac\varepsilon{n-m},\ldots,\frac\varepsilon{n-m}}_{n-m}\right)$
then this extracts min-free energy $\Delta F_0=k_B T\log(n/m)$ from $S$. Since $\varepsilon>0$ can be arbitrarily close to zero, and $\Delta F_0$ does not depend on $\varepsilon$, this changes the physical state of $S$ by an arbitrarily small amount. Thus, we obtain the following:\\
\begin{mdframed}[backgroundcolor=shadecolor,innertopmargin=\topskip]
\textbf{Main Result 3.} \emph{Consider some initial state $\rho_A$ and target state $\rho'_A$, both block-diagonal, such that $F(\rho_A)\geq F(\rho'_A)$. Using a work bit with energy gap $\Delta$ smaller than, but arbitrarily close to $F(\rho_A)-F(\rho'_A)$, we can implement the following transition with a thermal operation, which extracts work $\Delta$ without any fluctuations:
\[
   \rho_A\otimes \sigma_M\otimes \tau_S^{(m,n)}\otimes |g\rangle\langle g|_W\mapsto \sigma_{AMS}\otimes |e\rangle\langle e|_W.
\]
Here $\sigma_M={\rm Tr}_{AS}\sigma_{AMS}$ remains identical during the transformation, $\sigma_S=\tau_S^{(m,n,\varepsilon)}$, and $\sigma_A$ is as close to $\rho'_A$ as we like. This can be achieved for any choice of $\varepsilon>0$, as long as $n/m$ is large enough.}\\
\end{mdframed}

Since the state of the max entropy sink $S$ remains almost unchanged, the agent may measure the state of the sink after the transition, by checking whether its configuration is one of the $(n-m)$ basis states which have probability zero in the initial state $\tau^{(m,n)}_S$. With probability $1-\varepsilon$, this will yield the answer ``no'' and restore the original state $\tau^{(m,n)}_S$ due to state updating. However, even if $\varepsilon>0$ is very small, a large number of repetitions of the thermodynamic cycle will eventually lead to failure of the protocol.

In other words, the case of work extraction suffers from a deficit that is not present in the case of formation of a state: it admits only a weaker notion of cyclicity. An additional max entropy sink is needed, and its state is not reset with unit probability after every cycle. It is well-known that allowing small deviations from cyclicity can lead to quite implausible and unphysical effects like \emph{embezzling} of work~\cite{vanDam,Brandao}. Thus, we consider Main Result 3 as only a preliminary answer to the question of the ultimate limits of work extraction in the setup of this paper. Note that the authors of~\cite{Brandao} use a similar construction to dismiss the $F_\alpha$-conditions for $\alpha<0$.

The main source of the problem is to insist on producing the excited state $|e\rangle$ exactly. If we allow that this state is only obtained approximately, \emph{and} possibly correlated with the system $M$, then we obtain a valid alternative to Main Result 3 without any max entropy sink (simply by applying Main Result 1). The problem is that correlations between $W$ and $M$ may potentially compromise the working of the machine in further cycles. This leads to the question whether it can be ensured that $W$ remains uncorrelated with all other systems even if we drop the condition that it is in an exact eigenstate:\\

\begin{mdframed}[backgroundcolor=shadecolor,innertopmargin=\topskip]
\textbf{Open Problem.} Can we formulate a suitable version of Main Result 3 which allows the state of the work bit to be slightly mixed (dropping the max entropy sink), but which ensures nevertheless that it remains perfectly uncorrelated with all other systems (in particular $M$)?

This should be achieved in a way that allows to accumulate work over many extraction cycles without degrading its ``quality'' (fidelity with an eigenstate) and without the need for increasing resources or precision.\\
\end{mdframed}

We conjecture that the answer is ``yes'', and that it will lead to the same expression for the amount of work that can be extracted in the correlating scenario of this paper as suggested by Main Result 3, namely $F(\rho_A)-F(\rho'_A)$. A possible approach could be to adopt the methods of~\cite{Woods}, and to consider quasistatic ``near perfect'' work extraction. 

The authors of Ref.~\cite{Sparaciari} have recently shown that work can be extracted from passive states if the thermal machine $M$ is allowed to become correlated with the system. However, only work extraction on average was considered (not fluctuation-free single-shot work extraction like in this paper), the extracted work was only modelled implicitly, without the demand that unitaries preserve the total energy, and no heat bath (and thus background temperature) was considered. Thus the Helmholtz free energy $F$ plays no role in~\cite{Sparaciari}.

\subsection{Sketch of proof}
\label{SecProof}
Before discussing the role of coherence in Subsection~\ref{SubsecEngineering} below, we now give a self-contained sketch of the proof of the main results. It is mostly based on majorization theory and can be skipped by readers who are only interested in the physical discussion. All proof details can be found in the appendix.

Given any quantum system $X$ (which may itself be composed of several quantum systems), a \emph{thermal operation} on $X$ is a map $\rho_X\mapsto\rho'_X$ such that there exists a finite-dimensional system $B$ with
\[
   \rho'_X=\tr_B\left[U_{XB}\left(\rho_X\otimes\gamma_B\right) U_{XB}^\dagger\right],
\]
where $[U_{XB},H_X+H_B]=0$ for $H_X$ and $H_B$ the Hamiltonians of $X$ and $B$, and $\gamma_B=\exp(-\beta H_B)/Z$ is the Gibbs state, with $\beta=1/(k_B T)$ and $Z$ the partition function such that $\tr\,\gamma_B=1$ (the temperature $T$ is arbitrary but fixed). Our main results claim  that certain state transitions on composite systems are or are not possible via thermal operations. We make use of two technical simplifications to prove these results.

First, since we are only considering states that are block-diagonal in energy eigenbasis (except for Subsection~\ref{SubsecEngineering}), we can represent quantum states $\rho_X$ as probability vectors, $p_X\in\R^m$, where $m=\dim X$ is the dimension of $X$'s Hilbert space, and the entries of $p_X$ are the occupation probabilities of the (ordered) energy levels. A Hamiltonian $H_X$ can then be represented as a vector $H_X=(E_1,\ldots,E_m)$ with energies $E_i$, and it is for many purposes sufficient to consider only unitaries $U$ which correspond to permutations of entries of the probability vector, chosen such that $H_X$ is left invariant. See~\cite{Gour} and~\cite{Scharlau} for mathematical details.

Second, there is a well-known technique to reduce the study of (block-diagonal) thermal operations to the case where all Hamiltonians of all involved physical systems $Y$ are trivial, $H_Y=0$. This is achieved via an ``embedding map'' $\Gamma$ which, intuitively, reformulates the canonical state on some space as a microcanonical state on another space. This technique has been introduced in~\cite{Brandao} and used e.g.\ in \cite{Lostaglio} and~\cite{BeyondFreeEnergy} (the latter reference contains a summary in its Methods section).

In this simplified situation of trivial Hamiltonians and block-diagonal states, it can be shown that a state $p_X$ on some system $X$ can be transformed into another state $p'_X$ to arbitrary accuracy by a thermal operation if and only if $p_X$ majorizes~\cite{RuchMead,MOA} $p'_X$,
\[
   p_X\succ p'_X,
\]
which is shorthand for
\[
   \sum_{i=1}^k p^\downarrow_i \geq \sum_{i=1}^k p'^\downarrow_i\quad\mbox{for all }k=1,\ldots,m,
\]
where $p^\downarrow=(p_1^\downarrow,\ldots,p_m^\downarrow)$ denotes the reordering of $p$ in non-increasing order, i.e.\ $p^\downarrow_i=p_{\pi(i)}$ for some permutation $\pi$ such that $p^\downarrow_1\geq p^\downarrow_2\geq\ldots\geq p^\downarrow_m$.
This prescription does not yet take into account the possibility of having an additional catalyst $c_M$ as in Figure~\ref{fig_sketch}. Demanding, as in Subsections~\ref{SubsecKnown} and~\ref{SubsecExample}, that the catalyst remains uncorrelated with the system, we are led to the question under what conditions there exists some probability vector $c_M$ such that
\begin{equation}
   p_X\otimes c_M\succ p'_X\otimes c_M.
   \label{eqMajor}
\end{equation}
This question has been answered in~\cite{Klimesh} and~\cite{Turgut}: suppose that $p^\downarrow_X\neq p'^\downarrow_X$ and at least one of them does not contain zeros. Then there exists some state $c_M$ such that~(\ref{eqMajor}) holds if and only if $H_\alpha(p)<H_\alpha(p')$ for all $\alpha\in\R\setminus\{0\}$, and $H_{\rm Burg}(p)<H_{\rm Burg}(p')$, where the R\'enyi entropies $H_\alpha$~\cite{Renyi} and the Burg entropy $H_{\rm Burg}$~\cite{Burg} are defined as
\begin{eqnarray*}
H_\alpha(p)&=&\frac {{\rm sgn}(\alpha)}{1-\alpha}\log\sum_{i=1}^m p_i^\alpha \qquad (\alpha\in\R\setminus\{0,1\}),\\
H(p)\equiv H_1(p)&=&-\sum_{i=1}^m p_i\log p_i,\enspace H_{\rm Burg}(p)=\frac 1 m\sum_{i=1}^m \log p_i
\end{eqnarray*}
with $m=\dim X$ and ${\rm sgn}(\alpha)=+1$ if $\alpha>0$ and $-1$ if $\alpha<0$. Inverting the embedding $\Gamma$, allowing arbitrarily small errors in the production of the target state, and investing a tiny amount of extra work~\cite{Brandao} leads to condition~(\ref{egDeltaFAlpha}) for thermal transitions of the form~(\ref{eqCatalytic}), i.e.\ $\Delta F_\alpha\leq 0$ for all $\alpha$-free energies with $\alpha>0$.

The crucial step for establishing Main Results 1--3 is the following theorem that we prove in detail in the Appendix:\\

\begin{mdframed}[backgroundcolor=shadecolor,innertopmargin=\topskip]
\textbf{Main Theorem.} 
Let $p,p'\in\R^m$ be probability distributions with $p^\downarrow\neq p'^\downarrow$. Then there exists an extension $p'_{XY}$ of $p'\equiv p'_X$ such that
\begin{equation}
   p_X\otimes p'_Y \succ p'_{XY}
   \label{eqMajCorr}
\end{equation}
if and only if $H_0(p)\leq H_0(p')$ and $H(p)<H(p')$. Moreover, for every $\varepsilon>0$, we can choose $Y$ and $p'_{XY}$ such that the mutual information is $I(X:Y)\equiv S(p'_{XY}\|p'_X\otimes p'_Y)<\varepsilon$.\\
\end{mdframed}

The statement of this theorem uses the \emph{max entropy} (or \emph{Hartley entropy}) $H_0(p)=\log\#\{i\,\,|\,\, p_i\neq 0\}$, with its quantum version (also used in the main text) $S_0(\rho)=\log {\rm rank}(\rho)$, and it uses the notion of an ``extension'' of a probability distribution $p'$. To this end, we label the system on which $p'$ lives by $X$, and introduce another (discrete) system $Y$. An extension of $p'$ is then a joint probability distribution $p'_{XY}$ on the composite system $XY$ such that its marginal on $X$ equals $p'_X$. The mutual information $I(\bullet:\bullet)$ and relative entropy $S(\bullet\|\bullet)$ are defined in the Appendix. An interesting consequence is that, due to the Pinsker inequality~\cite{BZ},
\[
   \|p'_{XY}-p'_X\otimes p'_Y\|\leq\sqrt{I(X:Y)/2}<\sqrt{\varepsilon/2},
\]
where $\|p-q\|:=\frac 1 2 \sum_{i=1}^m |p_i-q_i|$ is the \emph{trace distance}, or \emph{variation distance}, which quantifies the distinguishability of $p$ and $q$~\cite{NC}. This means that $p'_{XY}$ can be as indistinguishable from a product state as we like, which is arguably the operationally strongest possible form of ``containing almost no correlations''.

Using the subadditivity~\cite{Aczel} of $H_0$ and $H=H_1$, it is very easy to see that $H_i(p)\leq H_i(p')$ for $i=0,1$ is necessary for the existence of some $p'_{XY}$ which satisfies~(\ref{eqMajCorr}). The hard part is to show that it is sufficient. To show this, we construct an explicit extension $p'_{XY}$ of $p'_X$ that satisfies~(\ref{eqMajCorr}). This is done in two steps. First, we introduce an auxiliary system $Y_1$ and an extension $p'_{XY_1}$ of $p'_X$ such that
\begin{eqnarray}
   H_\alpha(p_X\otimes p'_{Y_1})&<& H_\alpha(p'_{X Y_1})\quad\mbox{for all }\alpha\in\R\setminus\{0\},\nonumber\\
   H_{\rm Burg}(p_X\otimes p'_{Y_1})&<&H_{\rm Burg}(p'_{X Y_1}).\label{eqConditions}
\end{eqnarray}
The results of~\cite{Klimesh,Turgut} explained above will then guarantee that there is yet another auxiliary system $Y_2$ with a probability distribution $c_{Y_2}$ such that
\[
   p_X\otimes p'_{Y_1}\otimes c_{Y_2} \succ p'_{X Y_1}\otimes c_{Y_2},
\]
and we can simply define $Y:=Y_1 Y_2$ and $p'_{XY}:=p'_{X Y_1}\otimes c_{Y_2}$.

The extension $p'_{X Y_1}$ is explicitly defined in Figure~\ref{fig_graph}. While we can represent probability distributions $p_X$ on a system $X$ as vectors $p=(p_1,\ldots,p_m)\in\R^m$, we can similarly represent bipartite probability distributions $p_{XY_1}$ as matrices $p_{ij}$, like we do for $p'_{X Y_1}$ in Figure~\ref{fig_graph}. Summing over the rows resp.\ columns gives the marginals $p'_X=(p'_1,\ldots,p'_m)$ and $p'_{Y_1}$, which shows in particular that $p'_{XY_1}$ is indeed an extension of $p'_X$. We choose $Y_1$ to be $(n^2+n+1)$-dimensional, whereas $X$ is $m$-dimensional.

\begin{figure}
\begin{center}
\includegraphics[width=.4\textwidth]{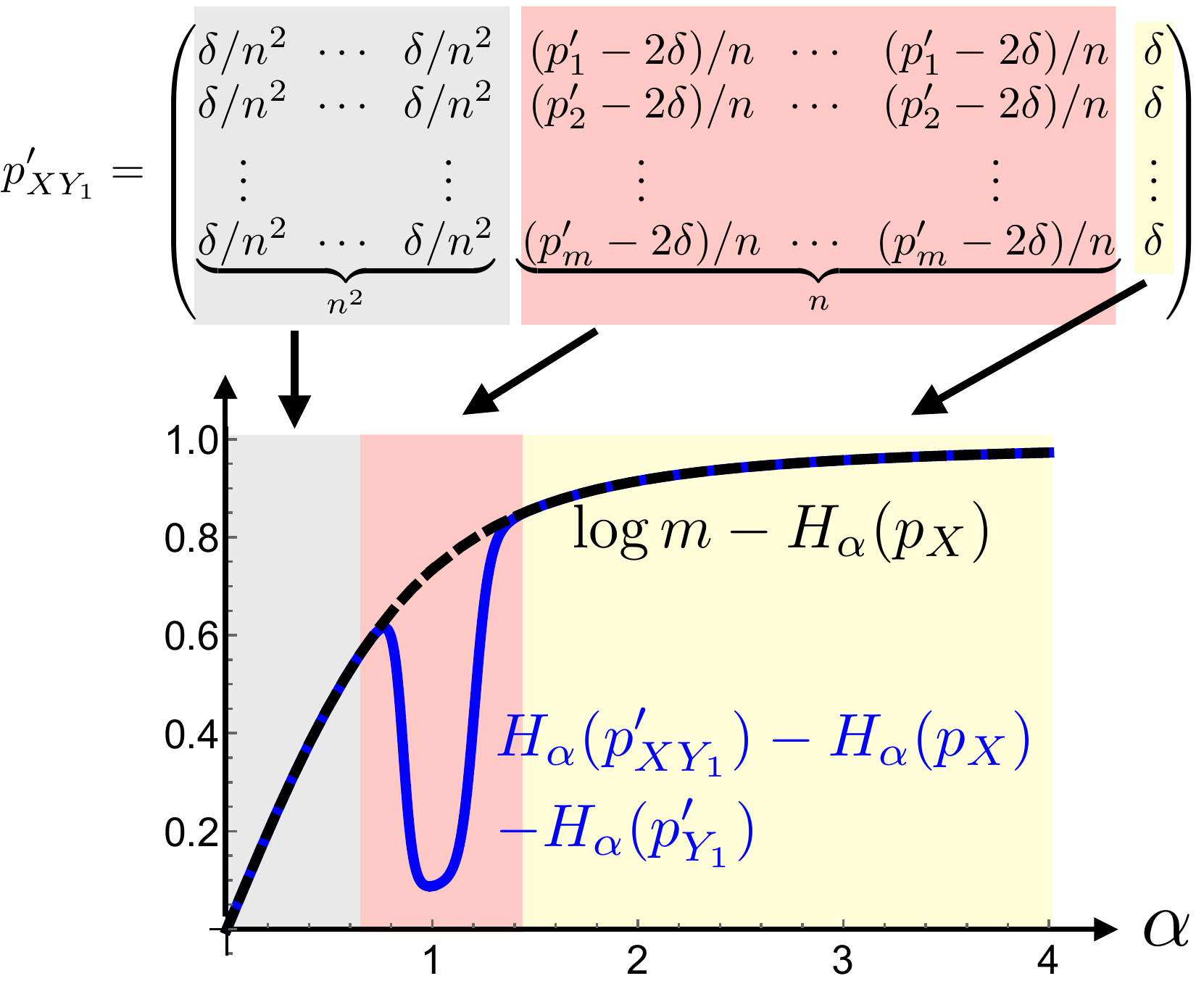}
\end{center}
\caption{\small The extension $p'_{XY_1}$ of $p'_X$ that is used in the main text to establish sufficiency of the entropy conditions in the Main Theorem. According to~(\ref{eqConditions}), the goal is to build an extension such that the blue curve (that is, the $\alpha$-R\'enyi entropy balance) attains only positive values. The plot is for $m=3$, $\delta=10^{-3}$, $p=p_X=(\frac{91}{100},\frac 1 {20},\frac 1 {25})$, $Y_1=\mathbb{R}^{n^2+n+1}$ with $n=10^{15}$ and $p'=p'_X=(\frac{17}{20},\frac{7}{50},\frac{1}{100})$. Since $H_\alpha(p)>H_\alpha(p')$ for $0<\alpha\leq \frac 1 3$, there does not exist $c_M$ such that~(\ref{eqMajor}) holds true, i.e.\ no standard catalytic thermal operation can transform $p$ into $p'$. Nevertheless, the transition can be achieved by a correlating catalytic thermal operation. The shaded colors show how different entries of $p'_{XY_1}$ are responsible for (the positivity of) different parts of the curve, as explained in the main text. In the limit $n\to\infty$, only positivity at $\alpha=1$, i.e.\ positive balance of Shannon entropy, remains as a necessary condition.}\label{fig_graph}
\end{figure}

Let us consider the special case that $p'_X$ does not contain zeros (implying $H_0(p)\leq H_0(p')$) and that $p_X\neq (\frac 1 m,\ldots,\frac 1 m)$. Suppose that $H(p)<H(p')$. We claim that for all $\alpha\neq 1$,
\[
   \lim_{n\to\infty} H_\alpha(p'_{X Y_1})-H_\alpha(p_X\otimes p'_{Y_1})=\log m-H_\alpha(p_X),
\]
which can be seen in Figure~\ref{fig_graph} by the fact that the left-hand side (the blue curve) approaches the right-hand side (the black dashed curve) for large $n$. In fact, the blue curve is monotonically increasing in $n$ towards the black curve. Since the maximal value of $H_\alpha(p_X)$ is $\log m$, and this is only attained at the uniform distribution, this shows that the blue curve attains strictly positive values away from $\alpha\in\{0,1\}$ if $n$ is large enough. According to the first condition in~(\ref{eqConditions}), this is exactly what we need to achieve.

We can understand why this happens by considering the different intervals of $\alpha$ separately. It turns out that the R\'enyi entropies $H_\alpha$ in the regime $\alpha>1$ are dominated by the largest elements of a probability distribution, which, in this case, are the $\delta$-entries (shaded yellow); all other entries do not contribute much to the value of $H_\alpha$. Since those entries are all equal, the expression $H_\alpha(p'_{XY_1})-H_\alpha(p'_{Y_1})$ reduces in the limit $n\to\infty$ to $\log m$. On the other hand, for $\alpha<1$, is is the \emph{smallest} entries of the probability distributions that matter, which are the $(\delta/n^2)$-entries (shaded grey), leading to the same conclusion. In fact, this intuition has been used in quantum information theory in the construction of counterexamples to certain versions of the so-called additivity conjecture~\cite{Hayden,Cubitt,Hastings}.

In contrast, for $\alpha=1$, the difference of entropies is \emph{constant} in $n$ and satisfies
\[
   \lim_{\delta\searrow 0} H_1(p'_{XY_1})-H_1(p_X\otimes p'_{Y_1})=H(p')-H(p).
\]
This explains why the blue curve in Figure~\ref{fig_graph} has an $n$-independent ``dip'' at $\alpha=1$: the value there differs in the limit from those at $\alpha<1$ and $\alpha>1$. Thus, the dip becomes very narrow as $n$ tends to infinity. The blue curve takes values at $\alpha\neq 1$ which are in the limit positive and \emph{independent} of the target state $p'_X$ and its extension $p'_{XY_1}$; it is only at $\alpha=1$ where the value depends on that state and its extension. If we choose $\delta$ small enough, we can enforce that the blue curve remains positive also around $\alpha=1$ if and only if $H(p')-H(p)>0$ --- that is, positivity of the standard Shannon entropy difference survives as the unique condition. One can show that the Burg entropy is related to the derivative of the blue curve at $\alpha=0$, and the second condition in~(\ref{eqConditions}) is automatically satisfied too, which establishes the first part of the Main Theorem. All remaining details of the proof are given in the Appendix.

Main Result 1 is then established by using an inverse of the embedding map $\Gamma$, as explained above. The proofs of Main Results 2 and 3 are very similar, except that some care has to be taken that all approximations (which are unavoidable due to the construction of $\Gamma$~\cite{Brandao}) are chosen without spoiling the purity of the work bit $W$. These results have thus independent (but very similar) proofs.

As we also show in the Appendix, a simple consequence of the result above is a resolution of an open problem in~\cite{MuellerPastena}: in the notation of that paper, it follows that c-trumping for $k=2$ is equivalent to c-trumping for $k\geq 3$.\\

\begin{mdframed}[backgroundcolor=shadecolor,innertopmargin=\topskip]
\textbf{Theorem~\ref{The2Trumping} (cf.\ Appendix).}
\emph{Let $p,q\in\R^m$ be probability distributions with $p\neq q$. Then there exist auxiliary systems $B,C$ and a bipartite distribution $r_{BC}$ such that}
\[
   p_A\otimes(r_B\otimes r_C)\succ q_A\otimes r_{BC}
\]
\emph{if and only if $H_0(p)\leq H_0(q)$ and $H(p)<H(q)$. Here, $r_B$ and $r_C$ denote the marginals of $r_{BC}$.}\\
\end{mdframed}
This also shows that $k=2$ systems are enough to use stochastic independence as a resource as described in~\cite{Lostaglio}, not only $k\geq 3$. We briefly comment on the relation between the present work and~\cite{MuellerPastena} after Theorem~\ref{The2Trumping} in the Appendix.

\subsection{Correlation and coherence?}
\label{SubsecEngineering}
So far, our discussion has focused on block-diagonal states, i.e.\ states that commute with the total Hamiltonian. In quantum thermodynamics, it is standard to consider this situation first, since transitions between states with coherence are much harder to characterize~\cite{BeyondFreeEnergy,LostaglioX,Cwiklinski}. In fact, the generic situation is that classification results for block-diagonal states fail to hold in the presence of coherence~\cite{Perarnau}, such as the equivalence of Gibbs-preserving and thermal operations~\cite{FOR}.

It is thus remarkable that the result of this paper has potentially a chance to hold in the presence of coherence as well:\\
\begin{mdframed}[backgroundcolor=shadecolor,innertopmargin=\topskip]
\textbf{Conjecture.} Main Result 1 remains true also in the case that $\rho_A$ and/or $\rho'_A$ are not block-diagonal, i.e.\ in the presence of quantum coherence.\\
\end{mdframed}
At first sight this may seem implausible: if, for example, $\rho'_A=\sigma_A$ is a pure state, $\sigma_{AM}$ must be a product state, and so the transition in Main Result 1 will simplify to
\begin{equation}
   \rho_A\otimes\sigma_M\mapsto \rho'_A\otimes\sigma_M,
   \label{eqCatalyticCorr2}
\end{equation}
which is just a standard catalytic thermal transition as discussed in Subsection~\ref{SubsecKnown}, subject to the family of ``second laws'' $\Delta F_\alpha\leq 0$ (not just $\Delta F\leq 0$). But this  ignores that we are in general only interested in producing the target state $\rho'_A$ \emph{approximately} (though to arbitrary accurary), such that $\sigma_A\approx\rho'_A$ may in general still be a mixed state, undermining the above counterargument.

If $\rho_A$ is incoherent and $\rho'_A$ is not, then a simple argument shows that transitions of the form~(\ref{eqCatalyticCorr2}) are impossible. Following~\cite{JanzingBeth}, define the \emph{quantum Fisher information} for a system with Hamiltonian $H$ and state $\rho$ as $I(\rho,H):={\rm tr}(\dot\rho \Delta_{\rho}^{-1} \dot\rho)$, where $\dot \rho:=i[\rho,H]$ and $\Delta_\rho X:=(\rho X + X \rho)/2$. Then $I(\rho,H)=0$ if and only if $\rho$ is incoherent. Moreover, $I$ is additive on tensor products, and $\rho\to\sigma$ by a thermal operation implies $I(\rho)\geq I(\sigma)$, since thermal operations are covariant. Applying these properties to~(\ref{eqCatalyticCorr2}) tells us that $I(\rho_A)\geq I(\rho'_A)$, i.e.\ if $\rho_A$ is block-diagonal then so is $\rho'_A$.

However, this kind of reasoning cannot be used to rule out Main Result 1: in general, it may hold $I(\sigma_A)+I(\sigma_M)> I(\sigma_{AM})$, and in this sense, correlations can increase the total amount of coherence as summed over all subsystems. This phenomenon is also at the heart of \r{A}berg's result~\cite{AbergCC} which gives us further evidence for the conjecture above. While \r{A}berg's setting is different from the one in this paper (his catalyst changes its state during every operation, and, in particular, is infinite-dimensional, thus exceeding the strict notion of cyclicity that we have adopted here --- similar comments apply to the improved results by Korzekwa et al.~\cite{Korzekwa}), his setup allows to ``broadcast'' coherence in some sense indefinitely catalytically, \emph{while correlating the catalyst with the systems on which it acts}, pretty much in the same way as in this paper. It has been noted that this comes at the prize of correlating the systems on which the catalyst successively acts~\cite{Bedkihal}. Therefore, the conjecture above blends into a series of questions about how to best use coherence catalytically~\cite{CirstoiuJennings}. We leave the resolution of this conjecture to future work.

\section{Conclusions}
It has been argued in~\cite{Brandao} that the Helmholtz free energy loses its role as the unique indicator of state transitions in small-scale thermodynamics. Instead, an infinite family of ``$\alpha$-free energies'' takes its place. It has been noted that this implies in particular that there is an inherent irreversibility at the nanoscale: while it takes $F_\infty(\rho)+k_B T \log Z$ to create a state $\rho$, only work $F_0(\rho)+k_B T \log Z$ can be extracted if one is given one copy of $\rho$, where in general $F_0<F_\infty$. But these results have been obtained under the assumption that the corresponding thermal machine remains uncorrelated from the systems on which it acts. In this paper, we have argued that this restriction can be lifted in many situations, and we have shown that this restores the distinguished role of the Helmholtz free energy $F$. Moreover, work extraction and formation at the free energy difference can be achieved without any fluctuations, up to a minor tweak in the extraction case.

Does this mean that we have restored reversibility at the nanoscale? Not quite. An interesting perspective to take is that this irreversibility has simply been shifted, from work to correlations. That is, while work cost and extractable work are now both equal to $F(\rho)$ (up to the Open Problem of Subsection~\ref{SubsecCorrExtr}), a new form of irreversibility has appeared: namely, initially uncorrelated systems (for example, $A$ and $M$) become correlated. It is interesting to see that this brings us closer to discussions of the founding days of thermodynamics: Boltzmann's H-theorem~\cite{Boltzmann}, for example, derives the non-decrease of entropy in a gas from the assumption that the velocities of molecules are initially uncorrelated (i.e.\ factorize), but they become correlated after a collision (``Sto\ss zahlansatz''). This introduces naturally an ``arrow of time'', and the fluctuation-free single-shot work formation and extraction in the present paper comes at the prize of introducing an analogous ``aging'' to the physical systems, with ``wrinkles'' given by correlations.

We emphasize that the results of this paper are not primarily meant as a criticism of earlier work. The point is not that it would be ``wrong'' to demand that the catalyst is returned uncorrelated (as in~(\ref{eqCatalyticCorr2})), but rather that the thermodynamic task of state conversion, when considered at the nanoscale, comes in two different versions: one version applicable to situations in which the machine acts on the \emph{same} system multiple times, such that the catalyst must be returned uncorrelated; and a second version, in which the machine acts on \emph{many different} quantum systems individually (and on each only once), in which case correlations are allowed to persist. The good (and arguably surprising) news of the present work is that the latter case is particularly simple to characterize, namely in terms of the free energy $F$ alone. The question of which version to choose depends entirely on the physical context.

The results of this paper open up a multitude of interesting open problems. First, does Main Result 1 remain true in the presence of coherence? Can we reformulate the work extraction result (Main Result 3) without a max entropy sink (Open Problem in Subsection~\ref{SubsecCorrExtr})? And finally, do machines that operate in this correlating-catalytic way have any realization in nature?

\section*{Acknowledgments}
I am grateful to Jonathan Oppenheim, Robert W.\ Spekkens, Henrik Wilming, and Mischa Woods for discussions, and in particular to Matteo Lostaglio for helpful discussions on coherence measures. This research was undertaken, in part, thanks to funding from the Canada Research Chairs program. This research was supported in part by Perimeter Institute for Theoretical Physics. Research at Perimeter Institute is supported by the Government of Canada through the Department of Innovation, Science and Economic Development Canada and by the Province of Ontario through the Ministry of Research, Innovation and Science.

\onecolumngrid
\section*{Appendix}
\subsection{Mathematical preliminaries}
In this paper, any ``probability distribution'' (or just ``distribution'') is assumed to be discrete, i.e.\ is a vector $p\in\R^m$ for some $m\in\N$ such that $p=(p_1,\ldots,p_m)$, $p_i\geq 0$, $\sum_{i=1}^m p_i=1$. We interpret it as a probability distribution on the discrete sample space $\{1,2,\ldots,m\}$, and we will usually denote the corresponding probability space by an uppercase letter like $A$, following quantum information terminology, writing $p\equiv p_A$. Given two probability spaces (``systems'') $A$ and $B$, we can consider the composite probability space $AB$ with a sample space that is the direct product of the two sample spaces. Independent product distributions will then be represented by vectors $p_A\otimes q_B$, and we can write joint probability distributions $q\equiv q_{AB}$ in matrix form, by collecting the probabilities $q_{AB}(i,j)$ into a table. Summing over the rows resp.\ columns of this matrix will give the marginal distributions on $A$ resp.\ $B$. We will sometimes slightly abuse notation and use uppercase letters like $A$ also as placeholders for the vector space $\R^m$ that contains its probability distributions, writing for example $p\in A$ instead of $p\in\R^m$. This improves clarity in cases where there is more than one system with sample space $\{1,2,\ldots,m\}$. Moreover, probability distributions will sometimes be called ``states'', again following quantum information terminology.

We define the notions of \emph{majorization}~\cite{MOA} and \emph{$\alpha$-R\'enyi entropies} $H_\alpha$ as well as \emph{Burg entropy} $H_{\rm Burg}$ as described in Subsection~\ref{SecProof}. A \emph{stochastic map} is a linear map $\Lambda:A\to A$ that maps probability distributions to probability distributions. A stochastic map is \emph{bistochastic} if it preserves the uniform distribution $\mu=(\frac 1 m,\ldots,\frac 1 m)\in\R^m$, i.e.\ $\Lambda(\mu)=\mu$. It is well-known that $p\succ q$ is equivalent to the existence of a bistochastic map $\Lambda$ such that $\Lambda(p)=q$~\cite{MOA}. Following~\cite{Nielsen,JonathanPlenio}, we say that a distribution $p_A$ \emph{trumps} another distribution $q_A$, denoted $p\succ_T q$, if there exists another (finite discrete) system $B$ and a distribution $c_B$ such that
\[
   p_A\otimes c_B \succ q_A\otimes c_B.
\]
As explained in Subsection~\ref{SecProof}, the relation $p\succ_T q$ for $p^\downarrow \neq q^\downarrow$ is equivalent to $H_\alpha(p)<H_\alpha(q)$ for all $\alpha\in\R\setminus\{0\}$ and $H_{\rm Burg}(p)<H_{\rm Burg}(q)$, which was proven in~\cite{Klimesh,Turgut}.

We use the trace norm (or trace distance~\cite{NC})
\[
   \|a\|:=\frac 1 2 \sum_{i=1}^m |a_i|,\qquad a=(a_1,\ldots,a_m)	\in\R^m.
\]
Stochastic maps $\Lambda$ do not increase the trace norm, i.e.\ $\|\Lambda(a)\|\leq \|a\|$ for all $a\in\R^m$.

Following~\cite{Brandao} (see also~\cite{vanErven}) we define the \emph{R\'enyi divergences}, or \emph{relative R\'enyi entropies}, for distributions $p,q\in\R^m$ as
\[
   S_\alpha(p\|q):=\frac{{\rm sgn}^+(\alpha)}{\alpha-1} \log \sum_{i=1}^m p_i^\alpha q_i^{1-\alpha} \qquad(\alpha\in\R\setminus\{0,1\}),
\]
where
\[
   {\rm sgn}^+(\alpha)=\left\{
      \begin{array}{cl}
      	   +1 & \mbox{if }\alpha\in[0,+\infty]\\
      	   -1 & \mbox{if }\alpha\in[-\infty,0).
      \end{array}
   \right.
\]
For $\alpha\in\{-\infty,0,1,\infty\}$, we use the definitions~\cite{Brandao}
\begin{eqnarray*}
   S_0(p\|q)=\lim_{\alpha\searrow 0} S_\alpha (p\|q)=-\log \sum_{i: p_i\neq 0} q_i, \qquad S_1(p\|q)\equiv S(p\|q) =\lim_{\alpha\to 1}S_\alpha(p\|q)=\sum_{i=1}^m p_i(\log p_i - \log q_i),\\
   S_\infty (p\|q)= \lim_{\alpha\to\infty} S_\alpha(p\|q)=\log \max_i \frac {p_i}{q_i}, \qquad
   S_{-\infty}(p\|q)=\lim_{\alpha\to-\infty} S_\alpha(p\|q)=S_\infty (q\|p).
\end{eqnarray*}

We will always assume that there is a fixed ``background inverse temperature'' $\beta>0$, and we will use the definition $k_B T:=1/\beta$, where we interpret $T$ as the temperature and $k_B$ as the Boltzmann constant. The $\alpha$-free energies $F_\alpha$ are defined as~\cite{Brandao}
\[
   F_\alpha(p):=-k_B T \log Z+ k_B T S_\alpha(p\|\gamma), \qquad F(p):=F_1(p),
\]
where $p\in A$ is any state, $Z\equiv Z_A=\sum_{i=1}^m \exp(-\beta E_i)$ is the partition function with $H_A=(E_1,\ldots E_m)$ the Hamiltonian (which, as described in Subsection~\ref{SecProof}, is now a vector with the energy levels as entries), and $\gamma=(\gamma_1,\ldots,\gamma_m)$ with $\gamma_i=\exp(-\beta E_i)/Z$ the thermal state (or Gibbs state).

Recall the definition of a thermal operation in Figure~\ref{fig_sketch}, but in the special case that the system $M$ is trivial, i.e.\ $AM=A$. If all states are block-diagonal, then we have a ``classical'' version of a thermal operation, acting effectively on classical probability distributions. If $p,q\in A$ are probability distributions, we can ask under what conditions a thermal operation can map the quantum state ${\rm diag}(p)$ to ${\rm diag}(q)$. This question was answered in~\cite{Janzing}, see also~\cite{HHO,Ruch,Scharlau}: this transition is possible to arbitrary accuracy if and only if there exists a stochastic map $\Lambda$ with
\[
   \Lambda(p)=q\quad\mbox{and}\quad \Lambda(\gamma_A)=\gamma_A
\]
(actually, in many but not all cases, the target state $q$ can be produced exactly by a thermal operation, i.e.\ with perfect accuracy, as discussed in~\cite{Scharlau}). Therefore, the existence of a thermal operation that maps one block-diagonal state to another can be shown by constructing a corresponding ``Gibbs-preserving'' stochastic map which maps the initial to the final distribution. The main result of~\cite{Brandao} was to give a criterion for the existence of a stochastic map $\Lambda$ with the above properties: basically (for details see~\cite{Brandao}), $F_\alpha(p)\geq F_\alpha(q)$ for all $\alpha$ is sufficient and necessary for the existence of such a map (we will not use this result directly in what follows).

\subsection{Results for trivial Hamiltonian}
As explained in the main text, we will in the following consider a particular family of bipartite probability distributions. For any given probability distribution $q\equiv q_A=(q_1,\ldots,q_m)\in\R^m$ with $q_i\neq 0$ for all $i$, we define the extension
\begin{equation}
	q_{AB}:=
	    \begin{pmatrix}
	       \delta & \delta/n^2 & \cdots & \delta/n^2 & (q_1-2\delta)/n & \cdots & (q_1-2\delta)/n \\
	       \delta & \delta/n^2 & \cdots & \delta/n^2 & (q_2-2\delta)/n & \cdots & (q_2-2\delta)/n \\
	       \vdots & \vdots &  & \vdots & \vdots &  & \vdots \\
	       \delta & \delta/n^2 & \sunderb{6.6em}{n^2} & \delta/n^2 & (q_m-2\delta)/n & \sunderb{13.4em}{n} & (q_m-2\delta)/n
	    \end{pmatrix}
	    \label{eqQAB}
\end{equation}
where $n\in\N$ and $0<\delta<\frac 1 2 \min_i q_i$. This is an $m\times(n^2+n+1)$-matrix with strictly positive entries which defines a joint probability distribution on $AB$. Summing over the rows shows that it has $q$ as its marginal on $A$. Its marginal on $B$ is
\[
   q_B=\left(
      m\delta,\underbrace{\frac{m\delta}{n^2},\ldots,\frac{m\delta}{n^2}}_{n^2},\underbrace{\frac{1-2m\delta}n,\ldots,\frac{1-2m\delta}n}_n
   \right).
\]
By direct computation, it turns out that the mutual information in $q_{AB}$ is independent of $n$:
\begin{equation}
   I(A:B)=S(q_{AB}\|q_A\otimes q_B)=\sum_{i=1}^m (q_i-2\delta)\log(q_i-2\delta)-\sum_{i=1}^m q_i \log q_i -2m\delta\log m -(1-2m\delta)\log(1-2m\delta),
   \label{eqMutualInformation}
\end{equation}
and we have in particular $\lim_{\delta\searrow 0}I(A:B)=0$.

\begin{lemma}
\label{LemBasic}
Let $p,q\in\R^m$ be probability distributions with full rank such that $H(p)<H(q)$. Then, for every $\varepsilon>0$, there exist $\delta>0$ with $\delta<\frac 1 2 \min_i q_i$ and $n\in\N$ such that $q_{AB}$ as defined in~(\ref{eqQAB}) satisfies
\[
   p_A\otimes q_B \succ_T q_{AB}
\]
as well as $I(A:B)\equiv S(q_{AB}\|q_A\otimes q_B)<\varepsilon$.
\end{lemma}
\begin{proof}
For $\alpha\in\R\cup\{-\infty,+\infty\}$, define the entropy difference
\[
   \Delta_n^{(\alpha)}:=H_{\alpha}(q_{AB})-H_{\alpha}(p_A)-H_{\alpha}(q_B).
\]
We claim that $\Delta_n^{(\alpha)}$ is everywhere continuous in $\alpha$. By definition this is true for all $\alpha\neq 0$; for $\alpha=0$, it follows from the fact that $p$ and $q$ both have full rank that $\lim_{\alpha\nearrow 0}\Delta_n^{(\alpha)}=\lim_{\alpha\searrow 0}\Delta_n^{(\alpha)}=\Delta_n^{(0)}=0$. 
Let us first compute this difference for $\alpha=1$. Defining $\eta(x):=-x\log x$ for $x\neq 0$ and $\eta(0):=0$, we get
\begin{eqnarray*}
\Delta_n^{(1)} &=& m \eta(\delta)+ mn^2\eta(\delta n^{-2})+\sum_{i=1}^m n\eta((q_i-2\delta)n^{-1})-H(p)-\eta(m\delta)-n^2 \eta(m\delta n^{-2})-n\eta((1-2m\delta)n^{-1})\\
&=&-\sum_{i=1}^m (q_i-2\delta)\log(q_i-2\delta)-H(p)+2m\delta\log m-\eta(1-2m\delta).
\end{eqnarray*}
All $n$-dependence miraculously cancels out, and we have
\[
   \lim_{\delta\searrow 0} \Delta_n^{(1)}=H(q)-H(p)>0.
\]
By continuity, positivity of $\Delta_n^{(1)}$ is ensured if $\delta$ is small enough. Furthermore, due to~(\ref{eqMutualInformation}), if $\delta$ is small enough, we will also have $I(A:B)<\varepsilon$ (note that $I(A:B)$ is in particular independent of $n$). We thus choose some $\delta\in\left(0,\frac 1 m\right)$ small enough for both and keep it fixed in all that follows. Consequently, $\Delta_n^{(1)}$ is constant in $n$ and positive, and $0<\delta<\frac 1 m$.

For finite $\alpha\not\in\{0,1\}$, we get
\begin{equation}
   \Delta_n^{(\alpha)}=-H_\alpha(p)+\frac{{\rm sgn}(\alpha)}{1-\alpha}\log\frac{m\delta^\alpha + m\delta^\alpha n^{2(1-\alpha)}+n^{1-\alpha}\sum_{i=1}^m (q_i-2\delta)^\alpha}{(m\delta)^\alpha +(m\delta)^\alpha n^{2(1-\alpha)}+(1-2m\delta)^\alpha n^{1-\alpha}}\qquad (\alpha\in\R\setminus\{0,1\}).
   \label{eqDelta}
\end{equation}
We claim that this expression is increasing in $n$, for every non-zero $\alpha\in\R\cup\{-\infty,\infty\}$. We have already shown this for $\alpha=1$, and now we will show it for all other $\alpha\not\in\{0,1\}$ by considering the following cases:
\begin{itemize}
	\item If $\alpha<0$ and $\alpha\neq-\infty$, then it is easy to see that $\Delta_n^{(\alpha)}$ is increasing in $n$ if and only if the fraction on the right-hand side of~(\ref{eqDelta}) is decreasing in $x:=n^{1-\alpha}$. In other words, we have a function
\begin{equation}
   f(x):=\frac{m\delta^\alpha +m\delta^\alpha x^2 + x\sum_{i=1}^m (q_i-2\delta)^\alpha}{(m\delta)^\alpha +(m\delta)^\alpha x^2 +(1-2m\delta)^\alpha x},
   \label{eqF}
\end{equation}
and we have to show that it is decreasing in $x$; note that we are only interested in $x\geq 1$, since $n^{1-\alpha}\geq n\geq 1$. To this end, we can simply look at the derivative
\[
   f'(x)=-\frac{(x^2-1)\delta^\alpha(m^\alpha\sum_{i=1}^m (q_i-2\delta)^\alpha - m(1-2m\delta)^\alpha)}{\left(\strut (m\delta)^\alpha+(m\delta)^\alpha x^2 +(1-2m\delta)^\alpha x\right)^2},
\]
and we see that it only remains to be shown that $m^\alpha \sum_{i=1}^m(q_i-2\delta)^\alpha\geq m(1-2m\delta)^\alpha$. Let $r_i:=(q_i-2\delta )/(1-2m\delta)$, then $r=(r_1,\ldots,r_m)$ is a probability distribution, and $H_\alpha(r)\leq -\log m$, which implies $\sum_{i=1}^m r_i^\alpha\geq m^{1-\alpha}$, and so
\[
   m^\alpha\sum_{i=1}^m (q_i-2\delta)^\alpha = m^\alpha (1-2m\delta)^\alpha \sum_{i=1}^m r_i^\alpha \geq m^{\alpha} (1-2m\delta)^\alpha m^{1-\alpha}=m(1-2m\delta)^\alpha
\]
which shows that $f'(x)\leq 0$ in the relevant interval for $x$, and we are done.
\item If $0<\alpha<1$, we can argue similarly, except that now the function $f$ in~(\ref{eqF}) has to be \emph{increasing} in $x=n^{1-\alpha}$. We can argue via the derivative exactly as above, but now $H_\alpha(r)\leq \log m$, hence $\sum_{i=1}^m r_i^\alpha\leq m^{1-\alpha}$, and therefore $m^\alpha \sum_{i=1}^m (q_i-2\delta)^\alpha\leq m(1-2m\delta)^\alpha$, which gives us the opposite sign, $f'(x)\geq 0$, as desired.
\item If $\alpha>1$, then the function $f$ in~(\ref{eqF}) also has to be increasing in $x=n^{1-\alpha}$, but since $1-\alpha<0$, we are now only interested in the interval $0<x<1$. On the one hand, we now have $\sum_{i=1}^m r_i^\alpha\geq m^{1-\alpha}$, which implies $m^\alpha \sum_{i=1}^m (q_i-2\delta)^\alpha\geq m(1-2m\delta)^\alpha$, but on the other hand, the factor $(x^2-1)$ in the derivative becomes negative, hence $f'(x)\geq 0$.
\item By continuity, $\Delta_n^{(\alpha)}$ must also be increasing for $\alpha\in\{-\infty,0,\infty\}$.
\end{itemize}
Since $\Delta_{n=1}^{(1)}>0$ and $\Delta_{n=1}^{(\alpha)}$ is continuous in $\alpha$, there exists some $\varepsilon>0$ such that $\Delta_{n=1}^{(\alpha)}>0$ for all $\alpha\in[1-\varepsilon,1+\varepsilon]$. But due to the monotonicity that we have just proven, it follows that $\Delta_n^{(\alpha)}>0$ for all $n\in\N$ and all $\alpha\in[1-\varepsilon,1+\varepsilon]$.

Now consider the interval $\alpha\in[1+\varepsilon,+\infty]$. On this interval, we have
\begin{equation}
   \lim_{n\to\infty} \Delta_n^{(\alpha)} = \log m-H_\alpha(p)>0,
   \label{eqLimitPositive}
\end{equation}
since $p$ cannot be the uniform distribution (due to $H(p)<H(q)$). For finite $\alpha\geq 1+\varepsilon$, this follows directly from~(\ref{eqDelta}), while for $\alpha=+\infty$, it follows from $H_\infty(q_{AB})=-\log\delta$ and $H_\infty(q_B)=-\log(m\delta)$ if $n$ is large enough.

Thus, on the interval $[1+\varepsilon,+\infty]$, the sequence of continuous functions $\Delta_n^{(\alpha)}$  converges pointwise to a strictly positive continuous function, namely $\log m-H_\alpha(p)$. Therefore, a version of Dini's theorem (see e.g.\ Lemma 6 in~\cite{MuellerPastena}) proves that there is some $N_+\in\N$ such that $\Delta_n^{(\alpha)}>0$ for all $n\geq N_+$ and all $\alpha\in [1+\varepsilon,+\infty]$.

Now consider the Burg entropy difference. A simple calculation yields
\begin{eqnarray*}
   \Delta_n^{\rm Burg}&:=&H_{\rm Burg}(q_{AB})-H_{\rm Burg}(p_A)-H_{\rm Burg}(q_B)\\
   &=&
   -H_{\rm Burg}(p)+\frac 1 {n^2+n+1}\left(
      \frac n m \sum_{i=1}^m \log(q_i-2\delta)-(n^2+1)\log m-n\log(1-2m\delta)
   \right).
\end{eqnarray*}
Thus, we obtain
\begin{equation}
   \lim_{n\to\infty}\Delta_n^{\rm Burg}=-\log m-H_{\rm Burg}(p)>0.
   \label{eqLimBurg}
\end{equation}
For all $\alpha\in\R$, define
\[
   \bar\Delta_n^{(\alpha)}:=\left\{
      \begin{array}{cl}
      	  \frac{1-\alpha}{|\alpha|}\Delta_n^{(\alpha)} & \mbox{if }\alpha\neq 0 \\
      	  \Delta_n^{\rm Burg} & \mbox{if }\alpha=0.
      \end{array}
   \right.
\]
Then $\bar\Delta_n^{(\alpha)}$ is continuous in $\alpha$ (in particular at $\alpha=0$). It is easy to verify that~(\ref{eqLimitPositive}) holds also true of $0<\alpha<1$; consequently, the $\bar\Delta_n^{(\alpha)}$ represent an increasing family of continuous functions on the compact interval $[0,1-\varepsilon]$ which converges to the continuous and strictly positive function (for $\alpha$ in that interval)
\[
   \lim_{n\to\infty}\bar\Delta_n^{(\alpha)}=\left\{
      \begin{array}{cl}
      	  \frac{1-\alpha}\alpha (\log m-H_\alpha(p)) & \mbox{if }\alpha>0 \\
      	  -\log m -H_{\rm Burg}(p) & \mbox{if }\alpha=0.
      \end{array}
   \right.
\]
Therefore, by Dini's theorem, there exists some $N_0\in\N$ such that $\bar\Delta_n^{(\alpha)}>0$ for all $n\geq N_0$ and all $\alpha\in[0,1-\varepsilon]$. But this implies that for all $n\geq N_0$, we have both $\Delta_n^{(\alpha)}>0$ for all $\alpha\in (0,1-\varepsilon]$ and $\Delta_n^{\rm Burg}>0$.

Now consider $\bar\Delta_n^{(\alpha)}$ on the interval $\alpha\in[-\infty,0]$. If $-\infty<\alpha<0$, then
\[
   \lim_{n\to\infty}\bar\Delta_n^{(\alpha)}=\frac{1-\alpha}{|\alpha|}\lim_{n\to\infty}\Delta_n^{(\alpha)}=\frac{1-\alpha}{|\alpha|}(-\log m -H_\alpha(p))>0.
\]
We also have
\[
   \bar\Delta_n^{(-\infty)}:=\lim_{\alpha\searrow -\infty}\Delta_n^{(\alpha)}=\Delta_n^{(-\infty)},
\]
and, if $n$ is large enough,
\[
   \Delta_n^{(-\infty)}=H_{-\infty}(q_{AB})-H_{-\infty}(q_B)-H_{-\infty}(p_A)=\log\frac\delta {n^2} - \log\frac{m\delta}{n^2}-H_{-\infty}(p)=-\log m -H_{-\infty}(p)>0
\]
since at least one entry of $p$ must be larger than $1/m$. Together with~(\ref{eqLimBurg}), this establishes that the $\bar\Delta_n^{(\alpha)}$ are a family of continuous functions on $[-\infty,0]$ that converge pointwise to a strictly positive continuous function. Again, by a version of Dini's theorem, it follows that there is some $N_-\in\N$ such that for all $n\geq N_-$, we have $\bar\Delta_n^{(\alpha)}>0$ and in particular $\Delta_n^{(\alpha)}>0$ for all $n\geq N_-$.

Thus, if we set $N:=\max\{N_-,N_0,N_+\}$, then for all $n\geq N$, we have that $\Delta_n^{(\alpha)}>0$ for all $\alpha\in[-\infty,+\infty]$ and also $\Delta_n^{\rm Burg}>0$. Therefore $p_A\otimes q_B\succ_T q_{AB}$.
\end{proof}

Lemma~\ref{LemBasic} remains true (under identical premises, and with the same form of catalyst) even if $p$ does not have full rank. We will now show this, but at the same time replace the trumping relation by majorization:
\begin{corollary}
\label{CorBasic}
Let $p,q\in\R^m$ be probability distributions such that $q$ (but not necessarily $p$) has full rank, and such that $H(p)<H(q)$. Then, for every $\varepsilon>0$ there is an extension $q_{AB}$ of $q=q_A$ with $I(A:B)\equiv S(q_{AB}\|q_A\otimes q_B)<\varepsilon$ such that
\[
   p_A\otimes q_B\succ q_{AB}.
\]
\end{corollary}
\begin{proof}
While $p$ does not necessarily have full rank, the distribution $p^{(\kappa)}\in\R^m$ does (for every $0<\kappa<1$), where $p_i^{(\kappa)}:=(1-\kappa)p_i+\kappa/m$. Since $H(p)<H(q)$ and $\lim_{\kappa\searrow 0} H(p^{(\kappa)})=H(p)$, there exists some $\kappa>0$ (smaller than one) such that $H(p^{(\kappa)})<H(q)$. Thus, we can apply Lemma~\ref{LemBasic} and get that there exists a system $C$ of suitable dimension and an extension $q'_{AC}$ of $q=q_A$ such that $p_A^{(\kappa)}\otimes q'_C\succ_T q'_{AC}$ and $S(q'_{AC}\|q'_A\otimes q'_C)<\varepsilon$. But $p\succ p^{(\kappa)}$, hence $p_A\otimes q'_C\succ p^{(\kappa)}_A\otimes q'_C$, therefore $p_A\otimes q'_C\succ_T p^{(\kappa)}_A\otimes q'_C$. Since the trumping relation is transitive, it follows that $p_A\otimes q'_C\succ_T q'_{AC}$. By definition of trumping, there exists yet another system $D$ of suitable dimension and a distribution $r_D$ such that $p_A\otimes q'_C\otimes r_D\succ q'_{AC}\otimes r_D$. Now we define $B$ to be the joint system $CD$, and $q_{AB}:=q'_{AC}\otimes r_D$, then $q_B=q'_C\otimes r_D$, and we have $p_A\otimes q_B\succ q_{AB}$. Furthermore,
\[
   S(q_{AB}\|q_A\otimes q_B)=S(q'_{AC}\otimes r_D\|q'_A\otimes q'_C\otimes r_D)=S(q'_{AC}\|q'_A\otimes q'_C)<\varepsilon.
\]
This completes the proof.
\end{proof}
This allows us to prove the main theorem of Subsection~\ref{SecProof}:
\begin{theorem}
\label{TheMaj}
Let $p,q\in\R^m$ be probability distributions with $p^\downarrow\neq q^\downarrow$. Then there exists an extension $q_{AB}$ of $q=q_A$ such that
\[
   p_A\otimes q_B\succ q_{AB}
\]
if and only if $H_0(p)\leq H_0(q)$ and $H(p)<H(q)$. Moreover, if these inequalities are satisfied, we can always choose $B$ and $q_{AB}$ such that $I(A:B)\equiv S(q_{AB}\|q_A\otimes q_B)<\varepsilon$, for any choice of $\varepsilon>0$.
\end{theorem}
\begin{proof}
\textbf{``Only if'' part.} If $p\neq q$ and $p_A\otimes q_B\succ q_{AB}$, then we get due to additivity, subadditivity and Schur concavity of $H_\alpha$ for $\alpha\in\{0,1\}$
\[
   H_\alpha(p_A)+H_\alpha(q_B)=H_\alpha(p_A\otimes q_B)\leq H_\alpha(q_{AB})\leq H_\alpha(q_A)+H_\alpha(q_B),
\]
thus $H_\alpha(p)\leq H_\alpha(q)$. This shows that $H_0(p)\leq H_0(q)$. Now consider the $\alpha=1$ case. While we also get $H(p)\leq H(q)$, equality (i.e.\ $H(p)=H(q)$) would entail that $H(q_{AB})=H(q_A)+H(q_B)$, which is only possible if $q_{AB}=q_A\otimes q_B$. But this would give us $p_A\otimes q_B\succ q_A\otimes q_B$, or $p_A\succ_T q_A$ for $p\neq q$, which implies that $H(p)<H(q)$.

\textbf{``If'' part.} We may assume without loss of generality that the entries of $p$ and $q$ are sorted in non-increasing order, i.e.\ $p_1\geq p_2\geq\ldots$ and $q_1\geq q_2\geq \ldots$. Since $q$ may not have full rank, we can ``split off all zeros'', by writing
\[
   q=\begin{pmatrix}\tilde q \\ 0 \\ \vdots \\ 0
   \end{pmatrix}
   \qquad\mbox{where }\tilde q\in\R^d\mbox{ has full rank, i.e.\ does not contain zeros, such that }d=2^{H_0(q)}\leq m.
\]
Since $H_0(p)\leq H_0(q)$, the distribution $p$ must contain at least as many zeros as $q$, such that we can also split off $(m-d)$ zeros, and write $p=(\tilde p,0,\ldots,0)^T$, where $\tilde p\in\R^d$. But then
\[
   H(\tilde p)=H(p)<H(q)=H(\tilde q),
\]
so Corollary~\ref{CorBasic} tells us that there is an extension $\tilde q_{AB}$ of $\tilde q$ such that $\tilde p_A\otimes \tilde q_B\succ \tilde q_{AB}$.  Moreover, no matter what $\varepsilon>0$ we have chosen, we can always choose $B$ and $\tilde q_{AB}$ such that $S(\tilde q_{AB}\|\tilde q_A\otimes \tilde q_B)<\varepsilon$. Using our matrix notation for bipartite distributions, denoting the dimension of the system $B$ by $k$, and using that adding a fixed number of zeros to two distributions does not change their majorization relation, we obtain
\[
   p_A\otimes \tilde q_B =\begin{pmatrix}
       \tilde p_1 \\ \vdots \\ \tilde p_d \\ 0 \\ \vdots \\ 0
   \end{pmatrix}\otimes\begin{pmatrix}
\tilde q_{B,1} \\ \vdots \\ \tilde q_{B,k}	
\end{pmatrix}=
\begin{pmatrix}
\tilde p_1 \tilde q_{B,1} & \tilde p_1 \tilde q_{B,2} & \cdots & \tilde p_1 \tilde q_{B,k} \\
\tilde p_2 \tilde q_{B,1} & \tilde p_2 \tilde q_{B,2} & \cdots & \tilde p_2 \tilde q_{B,k} \\
\vdots & \vdots & & \vdots \\
\tilde p_d \tilde q_{B,1} & \tilde p_d \tilde q_{B,2} & \cdots & \tilde p_d \tilde q_{B,k} \\
0 & 0 & \cdots & 0 \\
\vdots & \vdots & & \vdots \\
0 & 0 & \cdots & 0
\end{pmatrix} \succ
\begin{pmatrix}
& & & \\
& & \mathbf{\tilde q_{AB}} & \\
& & & \\
0 & 0 & \cdots & 0 \\
\vdots & \vdots & & \vdots \\
0 & 0 & \cdots & 0
\end{pmatrix}=:q'_{AB},
\]
where $\mathbf{\tilde q_{AB}}$ denotes $\tilde q_{AB}$ as a large matrix block. By summing over the rows, one sees that the marginal of $q'_{AB}$ on $A$ is $(\tilde q_1,\ldots, \tilde q_d,0,\ldots,0)^T=q_A$, and by summing over the columns, one obtains $\tilde q_B$ as the marginal on $B$. Thus, $q'_{AB}$ is the sought-for extension. Moreover, since the relative entropy does not change if common zero entries of both arguments are removed, we also have $S(q'_{AB}\|q'_A\otimes q'_B)=S(\tilde q_{AB}\|\tilde q_A\otimes \tilde q_B)<\varepsilon$.
\end{proof}
This result allows us to answer an open problem from~\cite{MuellerPastena}. There we have defined a notion of \emph{correlated trumping}: we say that $p$ \emph{c-trumps} $q$, denoted $p\succ_c q$, if there exists some $k\in\N_0$ and a $k$-partite distribution $r_{1,2,\ldots,k}$ such that
\begin{equation}
   p\otimes(r_1\otimes r_2\otimes\ldots\otimes r_k)\succ q\otimes r_{1,2,\ldots,k},
   \label{eqCTrumping}
\end{equation}
where $r_1,\ldots,r_k$ are the marginals of $r_{1,\ldots,k}$. In~\cite{MuellerPastena}, we have shows that $p\succ_c q$ for $p\neq q$ if and only if $H_0(p)\leq H_0(q)$ and $H(p)<H(q)$. We have also shown that we can always choose $k=3$, but we were not able to answer the question whether $k=2$ catalysts are always sufficient. Theorem~\ref{TheMaj} allows us to answer this question in the positive.
\begin{theorem}
\label{The2Trumping}
Let $p,q\in\R^m$ be probability distributions with $p\neq q$. Then there exist auxiliary systems $B,C$ and a bipartite distribution $r_{BC}$ such that
\[
   p_A\otimes(r_B\otimes r_C)\succ q_A\otimes r_{BC}
\]
if and only if $H_0(p)\leq H_0(q)$ and $H(p)<H(q)$. Here, $r_B$ and $r_C$ denote the marginals of $r_{BC}$.
\end{theorem}
\begin{proof}
The ``only if''-part of the proof is completely analogous to the corresponding part of the proof of Theorem~\ref{TheMaj} and thus omitted. For the ``if''-part, the premises $p\neq q$ and $H_0(p)\leq H_0(q)$ as well as $H(p)<H(q)$ imply, due to Theorem~\ref{TheMaj}, that there exists some auxiliary system $C$ and an extension $q_{AC}$ of $q=q_A$ such that $p_A\otimes q_C\succ q_{AC}$. Now introduce another system $B$ of the same dimension as $A$, and define a distribution $q_B$ which is just a copy of $q=q_A$. Then
\[
   p_A\otimes q_C\otimes q_B \succ q_{AC}\otimes q_B.
\]
Finally, since the majorization relation is permutation-invariant, we perform the swap of systems $A\leftrightarrow B$ on the right-hand side, and obtain
\[
   p_A\otimes (q_B\otimes q_C) \succ q_A\otimes q_{BC}
\]
(the left-hand side is simply a change of notation and not a physical swap). Thus, we can choose $r_{BC}:=q_{BC}$.
\end{proof}
Note that the results of~\cite{MuellerPastena}, i.e.\ the characterization of c-trumping (as defined in~(\ref{eqCTrumping})) via $H$ and $H_0$, is a strictly weaker result than the main majorization result of the present work, Theorem~\ref{TheMaj}. First, as the proof of Theorem~\ref{The2Trumping} above shows, the result of~\cite{MuellerPastena} can mathematically easily be obtained, and extended, from the results of the present paper. Second, Lemma 5 of~\cite{MuellerPastena} is a strictly weaker version of the present work's Theorem~\ref{TheMaj}, establishing sufficiency of the monotonicity of all $H_\alpha$, for $\alpha\geq 1$, for the existence of a correlating catalytic state transition (between full-rank states), while now we know that monotonicity of $H=H_1$ is enough. Regarding the thermodynamic version of~\cite{MuellerPastena} described in~\cite{Lostaglio}, c-trumping as in~(\ref{eqCTrumping}) can be physically interpreted as the irreversible use of $k$ auxiliary systems to admit a state transition $p\to q$ on the physical system of interest. That is, stochastic independence is used up as a ``fuel'' in a non-repeatable way. In contrast, the present paper describes a more natural thermodynamic scenario in which a single auxiliary system (that we can interpret as being part of a thermal machine) is used catalytically to implement state transitions on a single system. The auxiliary system can be used repeatedly on further copies of the system, which is arguably crucial for a thermodynamic cycle.

\subsection{Results for non-trivial Hamiltonians}
In this section, we will change our notation slightly, and call the auxiliary system $M$ (for ``thermal machine''), since $B$ is misleading in the thermodynamic context (it could be confused with the ``bath'').

Our main tool to transfer the results for trivial Hamiltonians to the case of non-trivial Hamiltonians will be a technique that has been introduced in~\cite{Brandao} and has also been applied in~\cite{Lostaglio}: the \emph{embedding map} $\Gamma_{\mathbf{d}}$. Given any ordered list of positive integers $\mathbf{d}=(d_1,\ldots,d_n)$, the  stochastic map $\Gamma_{\mathbf{d}}:\R^n\to \R^D$ is defined as
\[
   \Gamma_{\mathbf{d}}(p):=\bigoplus_{i=1}^np_i \mu_i
   =\left(
      \underbrace{\frac{p_1}{d_1},\ldots,\frac{p_1}{d_1}}_{d_1},
      \underbrace{\frac{p_2}{d_2},\ldots,\frac{p_2}{d_2}}_{d_2},\ldots,
      \underbrace{\frac{p_n}{d_n},\ldots,\frac{p_n}{d_n}}_{d_n}
   \right),
\]
where $\mu_i=\left(1/d_i,\ldots,1/d_i\right)\in\R^{d_i}$ is the uniform distribution in $d_i$ dimensions, and $D=\sum_{i=1}^n d_i$.
\begin{lemma}
\label{LemCanonMicro}
Let $A$ be a system with thermal distribution $\gamma_A$ that has only rational entries, i.e.\ that can be written in the form
\begin{equation}
   \gamma_A=\left( \frac{d_1}D,\frac{d_2}D,\ldots,\frac{d_n}D\right)\in\R^n.
   \label{eqGibbsRational}
\end{equation}
Then, for every $\alpha\in\R\cup\{-\infty,+\infty\}$, the $\alpha$-free energies of any $p_A$ are given by
\[
   F_\alpha(p_A)-F_\alpha(\gamma_A)\equiv kT\, S_\alpha(p_A\|\gamma_A)=kT\left({\rm sgn}^+(\alpha)\log D -H_\alpha(\Gamma_{\mathbf{d}}(p_A))\right),
\]
where $\mathbf{d}=(d_1,\ldots,d_n)$.
\end{lemma}
\begin{proof}
Simply evaluate the definition of $H_\alpha(\Gamma_{\mathbf{d}}(p_A))$ for the different cases of $\alpha$.
\end{proof}
In order to prove our main result, we need the following generalization and slight reformulation of Lemma 15 in~\cite{Brandao}.
\begin{lemma}
\label{LemApprox}
Let $r,r'\in\R^n$ be probability distributions with full rank (i.e.\ without any zero entries). Then there exists a stochastic map $\Phi:\R^n\to\R^n$ with $\Phi(r)=r'$ and
\[
   \|\Phi(p)-p\|\leq \max_j\left(1- \frac{r'_j}{r_j}\right)\qquad\mbox{for all probability distributions }p\in\R^n.
\]
In this sense, if $r\approx r'$ then $\Phi(p)\approx p$ for all distributions $p$.
\end{lemma}
\begin{proof}
The idea is to construct a map $\Phi$ that first ``shrinks'' the probability simplex, and then translates the shrunk simplex within the original simplex so that $r$ is mapped to $r'$. To this end, set $u(x):=x_1+\ldots+x_n$ for $x\in\R^n$, and for every $0\leq\lambda\leq 1$, define the ``shrinking map''
\[
   S_\lambda(x):=\lambda x+u(x)(1-\lambda)\mu \qquad (x\in\R^n,0\leq\lambda\leq 1),
\]
where $\mu\in\R^n$ is the uniform distribution. Finally, set
\[
   \Phi(x):=S_\lambda(x)-u(x)\left(S_\lambda(r)-r'\right)\qquad (x\in\R^n).
\]
It is easy to see that $\Phi$ preserves the normalization of probability distributions, i.e.\ $u(\Phi(x))=u(x)$ for all $x\in\R^n$. For $\Phi$ to be stochastic, it is thus necessary and sufficient that it maps the standard basis vectors $e_j=(0,\ldots,0,\underbrace{1}_j,0,\ldots,0)^T$ to vectors with non-negative entries. Since $\Phi(e_i)=\lambda e_i-\lambda r+r'$, we have $(\Phi(e_i))_i=\lambda(1- r_i)+r'_i$, which is non-negative since $r_i\leq 1$. On the other hand, for $(\Phi(e_i))_j=r'_j-\lambda r_j$ for $i\neq j$ to be non-negative, we need that $\lambda\leq r'_j/r_j$. Thus, if we define
\[
   \lambda:=\min_j \frac{r'_j}{r_j},
\]
the resulting map $\Phi$ will be stochastic. Since $\Phi(p)_i=\lambda p_i-\lambda r_i+r'_i$, we have
\[
   \|\Phi(p)-p\|=\frac 1 2 \sum_{i=1}^n |(\Phi(p))_i-p_i|\leq \frac 1 2 \sum_{i=1}^n \left(
      |\lambda p_i -p_i|+|\underbrace{r'_i-\lambda r_i}_{\geq 0}|
   \right)=1-\lambda
\]
for every probability distribution $p\in\R^n$, which completes the proof.
\end{proof}

\begin{theorem}
Consider a system $A$ with Hamiltonian $\H_A$ and two distributions $p_A$ and $q_A$.
Then, for every $\varepsilon>0$ there exists a distribution $q_A^\epsilon$ with $\|q_A^\varepsilon-q_A\|<\varepsilon$, an auxiliary system $M$ and an extension $q^\varepsilon_{AM}$ of $q^\varepsilon_A$ as well as a thermal operation $\mathcal{T}_\varepsilon$ with
\begin{equation}
   \mathcal{T}_\varepsilon(p_A\otimes q^\varepsilon_M)=q^\varepsilon_{AM}
   \label{eqCorrCat}
\end{equation}
if and only if $F(p_A)\geq F(q_A)$. Moreover, we can always choose the Hamiltonian on $M$ to be trivial, $\H_M=0$, and we can choose $M$ and $q_{AM}^\varepsilon$ such that $I(A:M)\equiv S(q_{AM}^\varepsilon\|q_A^\varepsilon\otimes q_M^\varepsilon)$ is (possibly nonzero but) as small as we like.

Note that the marginal on $M$ is exactly identical before and after the transformation, namely equal to $q_M^\varepsilon$.
\end{theorem}
\begin{proof}
We start with the \textbf{``only if''} part of the proof. Since the free energy $F$ is superadditive, decreasing under thermal operations, and additive, (\ref{eqCorrCat}) implies
\[
   F(q_A^\varepsilon)+F(q_M^\varepsilon)\leq F(q_{AM}^\varepsilon)\leq F(p_A\otimes q_M^\varepsilon)=F(p_A)+F(q_M^\varepsilon).
\]
Thus, for every $\varepsilon>0$ there is a distribution $q_A^\varepsilon$ which is $\varepsilon$-close to $q_A$ such that $F(q_A^\varepsilon)\leq F(p_A)$. Due to the continuity of $F$, it follows that $F(q_A)\leq F(p_A)$.

For the \textbf{``if''} direction, suppose that $p_A$ and $q_A$ are distributions with $F(p_A)\geq F(q_A)$, which is equivalent to $S(p_A\|\gamma_A)\geq S(q_A\|\gamma_A)$. First, consider the case that $q_A$ is the thermal state, $q_A=\gamma_A$. Then we can choose $M$ to be the trivial system, and $\mathcal{T}_\varepsilon$ can be the thermal operation that simply prepares the thermal state. Similarly, if $q_A=p_A$, then we can simply choose the identity map as our thermal operation. Let us now turn to the case $q_A\neq\gamma_A$ and $q_A\neq p_A$.

In general, the thermal distribution $\gamma_A$ will have non-rational entries and thus not be of the form~(\ref{eqGibbsRational}). However, since distributions with rational entries are dense in the set of all distributions, for every $\delta>0$, we can find another distribution $\gamma^{(\delta)}_A$ with all rational entries and $\max_j(1-\gamma^{(\delta)}_j/\gamma_j)<\delta$ as well as $\max_j(1-\gamma_j/\gamma^{(\delta)}_j)<\delta$ (just pick $\gamma^{(\delta)}$ close enough to $\gamma$). Due to Lemma~\ref{LemApprox}, there exists a stochastic map $\Phi:\R^n\to\R^n$ such that $\Phi(\gamma)=\gamma^{(\delta)}$ and $\|\Phi(s)-s\|<\delta$ for all distributions $s\in\R^n$, and there also exists a stochastic map $\bar\Phi:\R^n\to\R^n$ with $\bar\Phi(\gamma^{(\delta)})=\gamma$ and $\|\bar\Phi(s)-s\|<\delta$ for all distributions $s\in\R^n$. Writing
\[
   \gamma^{(\delta)}=\left(\frac{d_1^{(\delta)}}{D_\delta},\frac{d_2^{(\delta)}}{D_\delta},\ldots,\frac{d_n^{(\delta)}}{D_\delta}\right),\qquad
   \mathbf{d}^\delta:=\left(d_1^{(\delta)},d_2^{(\delta)},\ldots,d_n^{(\delta)}\right)\in\N^n,
\]
we obtain a corresponding embedding map $\Gamma_{\mathbf{d}^\delta}=:\Gamma_\delta$ that we will use shortly.

But before doing so, define $q^{(\varepsilon)}_A:=(1-\frac{\varepsilon} 2 )q_A+\frac\varepsilon 2 \gamma_A$ for every $0<\varepsilon<1$. It follows that $\|q_A^{(\varepsilon)}-q_A\|\leq \frac\varepsilon 2$. Due to the convexity of the relative entropy, we have
\[
   S\left(q_A^{(\varepsilon)}\|\gamma_A\right)\leq \left(1-\frac\varepsilon 2\right)S(q_A\|\gamma_A)+\frac \varepsilon 2 S(\gamma_A\|\gamma_A)=\left(1-\frac \varepsilon 2 \right)S(q_A\|\gamma_A)\leq\left(1-\frac\varepsilon 2\right)S(p_A\|\gamma_A).
\]
Let $p_A^{(\delta)}:=\Phi_A(p_A)$, then $\|p_A^{(\delta)}-p_A\|<\delta$. Since $\lim_{\delta\searrow 0}S\left(q_A^{(\varepsilon)}\|\gamma_A^{(\delta)}\right)=S(q_A^{(\varepsilon)}\|\gamma_A)$ and $\lim_{\delta\searrow 0} S\left(p_A^{(\delta)}\|\gamma_A^{(\delta)}\right)=S(p_A\|\gamma_A)$ due to continuity, we can pick $\delta>0$ small enough such that
\[
   S\left(q_A^{(\varepsilon)}\|\gamma_A^{(\delta)}\right)<S\left(p_A^{(\delta)}\|\gamma_A^{(\delta)}\right).
\]
In the following, let us assume that, for any choice of $\varepsilon>0$, we have chosen $\delta>0$ small enough for this to be satisfied, and by doing so we also make sure that $\delta<\varepsilon/2$. Then Lemma~\ref{LemCanonMicro} shows that
\[
   H\left(\Gamma_\delta(p_A^{(\delta)})\right)< H\left(\Gamma_\delta(q_A^{(\varepsilon)})\right).
\]
Since $\gamma_A$ has full rank, so does $q_A^{(\varepsilon)}$, and thus $\Gamma_\delta(q_A^{(\varepsilon)})$ has full rank too, hence $H_0\left(\Gamma_\delta(p_A^{(\delta)})\right)\leq H_0\left(\Gamma_\delta(q_A^{(\varepsilon)})\right)$.
Denoting the $D_\delta$-dimensional system by $A'$, such that $\Gamma_\delta$ is a map from $A$ to $A'$, Theorem~\ref{TheMaj} tells us that there exists a distribution $r^{(\varepsilon)}_{A'M}$ on $A'M$ (recall that $\delta$ depends on the choice of $\varepsilon$) such that $r^{(\varepsilon)}_{A'}=\Gamma_\delta(q_A^{(\varepsilon)})$ and $\Gamma_\delta(p_A^{(\delta)})\otimes r_M^{(\varepsilon)}\succ r_{A'M}^{(\varepsilon)}$. Moreover, for any choice of $\kappa>0$, we can choose this state such that $I(A':M)\equiv S(r_{A'M}^{(\varepsilon)}\|r_{A'}^{(\varepsilon)}\otimes r_M^{(\varepsilon)})<\kappa$. Therefore, there exists a bistochastic map $\Lambda_\varepsilon:A'\otimes M\to A'\otimes M$ (i.e.\ a stochastic map with $\Lambda_\varepsilon(\mu_{A'}\otimes\mu_M)=\mu_{A'}\otimes\mu_M$) such that
\[
   \Lambda_\varepsilon\left(\Gamma_\delta(p_A^{(\delta)})\otimes r_M^{(\varepsilon)}\right)=r_{A'M}^{(\varepsilon)}.
\]
Let us define a stochastic map $\bar\Gamma_\delta:\R^{D_\delta}\to\R^n$ which is a pseudo-inverse of $\Gamma_\delta$ via
\[
   \bar\Gamma_\delta(x):=\left(\sum_{i=1}^{d_1^\delta} x_i,\sum_{i=d_1^\delta}^{d_1^\delta+d_2^\delta}x_i,\ldots,\sum_{i=d_1^\delta+\ldots+d_{n-1}^\delta}^{D_\delta} x_i\right)\qquad (x=(x_1,\ldots,x_n)\in\R^{D_\delta}),
\]
then we have $\bar\Gamma_\delta\circ\Gamma_\delta=\id_A$. Furthermore, define a linear map $\mathcal{T}_\varepsilon:A\otimes M \to A\otimes M$ via
\[
   \mathcal{T}_\varepsilon:=\left(\bar\Phi_A\otimes \id_M\right)\circ \left(\bar\Gamma_\delta\otimes\id_M\right)\circ \Lambda_\varepsilon \circ \left(\Gamma_\delta\otimes\id_M\right)\circ\left(\Phi_A\otimes\id_M\right).
\]
As a composition of stochastic maps, $\mathcal{T}_\varepsilon$ is stochastic, i.e.\ maps probability distributions to probability distributions. If we equip $M$ with the trivial Hamiltonian $\H_M=0$, the thermal distribution on $A\otimes M$ is $\gamma_A\otimes\mu_M$. Using some previous identities, it is easy to see that
\[
   \mathcal{T}_\varepsilon(\gamma_A\otimes\mu_M)=\gamma_A\otimes\mu_M,
\]
hence $\mathcal{T}_\varepsilon$ is a thermal operation. Similarly, we obtain
\[
   \mathcal{T}_\varepsilon\left(p_A\otimes r_M^{(\varepsilon)}\right)=\left(\bar\Phi_A\otimes\id_M\right)\circ\left(\bar\Gamma_\delta \otimes \id_M\right)(r_{A'M}^{(\varepsilon)})=:s_{AM}^{(\varepsilon)}.
\]
From this equation, we see that the marginal on $M$ is $s_M^{(\varepsilon)}=r_M^{(\varepsilon)}$. The marginal on $A$ is
\[
   s_A^{(\varepsilon)}=\bar\Phi_A\left(\bar\Gamma_\delta\left(r_{A'}^{(\varepsilon)}\right)\right)
   =\bar\Phi_A\left(\bar\Gamma_\delta\left(\Gamma_\delta\left(q_A^{(\varepsilon)}\right)\right)\right)=\bar\Phi_A\left(q_A^{(\varepsilon)}\right)=:q_A^{(\varepsilon,\delta)},
\]
and this distribution is $\varepsilon$-close to $q_A$:
\[
   \|q_A^{(\varepsilon,\delta)}-q_A\|\leq \|q_A^{(\varepsilon,\delta)}-q_A^{(\varepsilon)}\|+\|q_A^{(\varepsilon)}-q_A\|<\delta+\frac\varepsilon 2<\varepsilon.
\]
Thus, we can set $q_{AM}^\varepsilon:=s_{AM}^{(\varepsilon)}$. To prove the final part of the claim, recall the the relative entropy is non-increasing under stochastic maps, hence
\begin{eqnarray*}
S\left(q_{AM}^\varepsilon\|q_A^\varepsilon\otimes q_M^\varepsilon\right)&=& S\left(s_{AM}^{(\varepsilon)}\|s_A^{(\varepsilon)}\otimes s_M^{(\varepsilon)}\right)=S\left(\left(\bar\Phi_A\otimes\id_M\right)\circ \left(\bar\Gamma_\delta \otimes \id_M\right)(r_{A'M}^{(\varepsilon)})\| \bar\Phi_A\left(\bar\Gamma_\delta(r_{A'}^{(\varepsilon)})\right)\otimes r_M^{(\varepsilon)}\right)\\
&=&S\left(\left(\bar\Phi_A\otimes\id_M\right)\circ \left(\bar\Gamma_\delta \otimes \id_M\right)(r_{A'M}^{(\varepsilon)})\| \left(\bar\Phi_A\otimes\id_M\right)\circ \left(\bar\Gamma_\delta \otimes \id_M\right)(r_{A'}^{(\varepsilon)}\otimes r_M^{(\varepsilon)})
\right)\\
&\leq & S(r_{A'M}^{(\varepsilon)}\|r_{A'}^{(\varepsilon)}\otimes r_M^{(\varepsilon)})<\kappa.
\end{eqnarray*}
\vskip -2em
\end{proof}
In order to talk about work extraction, we need to introduce work bits. A \emph{work bit system $W$ with energy gap $\Delta\in\R$} is a binary system $W=\R^2$ with Hamiltonian $\H_W=(0,\Delta)$. We will usually consider situations where $\Delta\geq 0$, and we will in particular allow that $\Delta=0$, i.e.\ that $\H_W$ is degenerate.
\begin{theorem}[Performing work on the system]
\label{ThePerformingWork}
Consider a system $A$ with Hamiltonian $\H_A$ and two distributions $p_A$ and $q_A$ such that $F(p_A)\leq F(q_A)$. Suppose we would like to transform $p_A$ approximately into $q_A$ with the help of spending some energy $\Delta\geq 0$. Then, for every $\delta,\varepsilon>0$, we can find some $\Delta<F(q_A)-F(p_A)+\delta$ and a thermal operation $\mathcal{T}_{\delta,\varepsilon}$ such that
\begin{equation}
   \mathcal{T}_{\delta,\varepsilon}\left(p_A\otimes (0,1)_W\otimes q_M^{\delta,\varepsilon}\right)=q_{AM}^{\delta,\varepsilon}\otimes (1,0)_W,
\end{equation}
where $\|q_A^{\delta,\varepsilon}-q_A\|<\varepsilon$, $q_{AM}^{\delta,\varepsilon}$ is a suitable extension of $q_A^{\delta,\varepsilon}$, and $W$ is a work bit with energy gap $\Delta$. In particular, the work bit transforms from a pure excited state $(0,1)_W$ to a pure ground state $(1,0)_W$ and does not become correlated with $AM$.
\end{theorem}
\begin{proof}
We use the convention $\beta:=1/(k_B T)$. Consider the thermal state $\gamma_W$ of the work bit:
\[
   \gamma_W=\frac 1 {1+e^{-\beta\Delta}}\left(1,e^{-\beta\Delta}\right).
\]
The set of $\Delta$ for which $e^{-\beta\Delta}$ is rational is dense in $\R$. Thus, for every $\delta>0$, we can find some $\Delta$ with $F(q_A)-F(p_A)<\Delta<F(q_A)-F(p_A)+\delta$ such that $e^{-\beta\Delta}$ is rational. We pick one arbitrarily; consequently, $\gamma_W$ has rational entries.

In the following, we will suppress the dependence from $\delta$ for notational simplicity; it will however be explicitly denoted in the statement of the theorem.

Now let $q^{(\varepsilon)}_A:=(1-\frac{\varepsilon} 2 )q_A+\frac\varepsilon 2 \gamma_A$, then $q_A^{(\varepsilon)}$ has full rank, and
\begin{eqnarray}
S\left(q_A^{(\varepsilon)}\otimes (1,0)_W\|\gamma_A\otimes\gamma_W\right)&=& S\left(q_A^{(\varepsilon)}\|\gamma_A\right)+S\left((1,0)_W\|\gamma_W\right)
\leq S(q_A\|\gamma_A)+\beta \underbrace{F(\left( 1,0)_W\right)}_0-\beta F(\gamma_W)\nonumber\\
&=& \beta F(q_A)-\beta F(\gamma_A)-\beta F(\gamma_W)
<\beta\underbrace{\Delta}_{F((0,1)_W)}+\beta F(p_A)-\beta F(\gamma_A)-\beta F(\gamma_W)
\nonumber\\
&=& S(p_A\|\gamma_A)+S\left((0,1)_W\|\gamma_W\right) = S\left(p_A\otimes (0,1)_W\|\gamma_A\otimes \gamma_W\right).
\label{eqRelEntDiff}
\end{eqnarray}
In general, the thermal distribution $\gamma_A$ will have non-rational entries and thus not be of the form~(\ref{eqGibbsRational}). However, since distributions with rational entries are dense in the set of all distributions, for every $\kappa>0$, we can find another distribution $\gamma^{(\kappa)}_A$ with all rational entries and $\max_j(1-\gamma^{(\kappa)}_j/\gamma_j)<\kappa$ as well as $\max_j(1-\gamma_j/\gamma^{(\kappa)}_j)<\kappa$ (just pick $\gamma^{(\kappa)}$ close enough to $\gamma$). Due to Lemma~\ref{LemApprox}, there exists a stochastic map $\Phi:A\to A$ such that $\Phi(\gamma)=\gamma^{(\kappa)}$ and $\|\Phi(s)-s\|<\kappa$ for all distributions $s\in\R^n$, and there also exists a stochastic map $\bar\Phi:A\to A$ with $\bar\Phi(\gamma^{(\kappa)})=\gamma$ and $\|\bar\Phi(s)-s\|<\kappa$ for all distributions $s\in\R^n$. Set $p_A^{(\kappa)}:=\Phi(p_A)$, then $\|p_A-p_A^{(\kappa)}\|<\kappa$. Due to the continuity of the relative entropy, we can find some $0<\kappa<\varepsilon/2$ that is small enough such that the inequality of~(\ref{eqRelEntDiff}) is still true if $\gamma_A$ is replaced by $\gamma_A^{(\kappa)}$, and if $p_A$ is replaced by $p_A^{(\kappa)}$:
\begin{equation}
   S\left(q_A^{(\varepsilon)}\otimes (1,0)_W\|\gamma_A^{(\kappa)}\otimes \gamma_W\right)< S\left(p_A^{(\kappa)}\otimes (0,1)_W\|\gamma_A^{(\kappa)}\otimes \gamma_W\right).
   \label{eqEntropyUp}
\end{equation}
Since both $\gamma_A^{(\kappa)}$ and $\gamma_W$ have all rational entries, we can write
\[
   \gamma_A^{(\kappa)}\otimes\gamma_W=\left(\frac{d_1^{(\kappa)}}{D_\kappa},\frac{d_2^{(\kappa)}}{D_\kappa},\ldots,\frac{d_n^{(\kappa)}}{D_\kappa}\right),\qquad
   \mathbf{d}^\kappa:=\left(d_1^{(\kappa)},d_2^{(\kappa)},\ldots,d_n^{(\kappa)}\right),\qquad \mbox{all }d_i^{(\kappa)}\in\N,
\]
and obtain a corresponding embedding map $\Gamma_{\mathbf{d}^\kappa}=:\Gamma_\kappa$. Due to Lemma~\ref{LemCanonMicro}, we get
\begin{equation}
   H\left(\Gamma_\kappa\left(p_A^{(\kappa)}\otimes (0,1)_W\right)\right)< H\left(\Gamma_\kappa\left(q_A^{(\varepsilon)}\otimes (1,0)_W\right)\right).
   \label{eqH1Gamma}
\end{equation}
Let us now check the balance of R\'enyi divergence $S_0$, in analogy to~(\ref{eqEntropyUp}). Using $Z_W:=1+e^{-\beta\Delta}$, we get
\begin{eqnarray*}
S_0\left(q_A^{(\varepsilon)}\otimes (1,0)_W\|\gamma_A^{(\kappa)}\otimes \gamma_W\right)&=&\underbrace{S_0\left(q_A^{(\varepsilon)}\|\gamma_A^{(\kappa)}\right)}_0 + S_0\left( (1,0)_W\|\gamma_W\right)=\log Z_W,\\
S_0\left(p_A^{(\kappa)}\otimes (0,1)_W\|\gamma_A^{(\kappa)}\otimes\gamma_W\right)&\geq& S_0\left((0,1)_W\|\gamma_W\right)=\log Z_W+\beta\Delta\geq \log Z_W
\end{eqnarray*}
since $\Delta\geq 0$ (note that this is where it becomes important that we talk about performing work on the system, not about extracting work from the system). Using Lemma~\ref{LemCanonMicro} again, we obtain
\begin{equation}
   H_0\left(\Gamma_\kappa\left(p_A^{(\kappa)}\otimes (0,1)_W\right)\right)\leq H_0\left(\Gamma_\kappa\left(q_A^{(\varepsilon)}\otimes (1,0)_W\right)\right).
   \label{eqH2Gamma}
\end{equation}
Now we can apply Theorem~\ref{TheMaj}: denoting the image of $AW$ under $\Gamma_\kappa$ by $(AW)'$, it follows from~(\ref{eqH1Gamma}) and~(\ref{eqH2Gamma}) that there exists a distribution $r_{(AW)'M}^{(\varepsilon)}$ with $r_{(AW)'}^{(\varepsilon)}=\Gamma_\kappa \left(q_A^{(\varepsilon)}\otimes(1,0)_W\right)$ and $\Gamma_\kappa\left(p_A^{(\kappa)}\otimes (0,1)_W\right)\otimes r_M^{(\varepsilon)}\succ r_{(AW)'M}^{(\varepsilon)}$. Therefore, there exists a bistochastic map $\Lambda_{\varepsilon}:(AW)'M\to(AW)'M$ such that
\[
   \Lambda_{\varepsilon}\left(\Gamma_\kappa\left(p_A^{(\kappa)}\otimes (0,1)_W\right)\otimes r_M^{(\varepsilon)}\right)=r_{(AW)'M}^{(\varepsilon)}.
\]
We define a stochastic map $\bar\Gamma_\kappa:(AW)'\to AW$ which is a pseudo-inverse of $\Gamma_\kappa$ via
\[
   \bar\Gamma_\kappa(x):=\left(\sum_{i=1}^{d_1^\kappa} x_i,\sum_{i=d_1^\kappa}^{d_1^\kappa+d_2^\kappa}x_i,\ldots,\sum_{i=d_1^\kappa+\ldots+d_{n-1}^\kappa}^{D_\kappa} x_i\right)\qquad (x=(x_1,\ldots,x_n)\in\R^{D_\kappa}),
\]
so that we get $\bar\Gamma_\kappa\circ\Gamma_\kappa=\id_{AW}$. Now we define a linear map
\[
   \mathcal{T}_{\varepsilon}:=\left(\bar\Phi_A\otimes \id_{WM}\right)\circ \left(\bar\Gamma_\kappa\otimes\id_M\right)\circ \Lambda_{\varepsilon}\circ \left(\Gamma_\kappa \otimes\id_M\right)\circ\left(\Phi_A\otimes\id_{WM}\right).
\]
It is straightforward to check that $\mathcal{T}_{\varepsilon}$ maps the thermal state $\gamma_A\otimes\gamma_W\otimes\mu_M$ of $AWM$ onto itself, hence it is a thermal operation. Furthermore,
\[
   \mathcal{T}_{\varepsilon}\left(p_A\otimes (0,1)_W\otimes r_M^{(\varepsilon)}\right)=\left(\bar\Phi_A\otimes \id_{WM}\right)\circ \left(\bar\Gamma_\kappa\otimes \id_M\right)\left(r_{(AW)'M}^{(\varepsilon)}\right)=:s_{AWM}^{(\varepsilon)}.
\]
Thus $s_M^{(\varepsilon)}=r_M^{(\varepsilon)}$, and
\[
   s_{AW}^{(\varepsilon)}=\left(\bar\Phi_A\otimes \id_W\right)\circ\left(\bar\Gamma_\kappa\left(r_{(AW)'}^{(\varepsilon)}\right)\right)=\bar\Phi_A\otimes\id_W\left(\bar\Gamma_\kappa\left(\Gamma_\kappa\left(q_A^{(\varepsilon)}\otimes (1,0)_W\right)\right)\right)
   =\bar\Phi_A\left(q_A^{(\varepsilon)}\right)\otimes (1,0)_W.
\]
Since pure states are always uncorrelated with other systems, we obtain $s_{AWM}^{(\varepsilon)}=(1,0)_W\otimes s_{AM}^{(\varepsilon)}$.
We also get
\[
   \left\|s_A^{(\varepsilon)}-q_A\right\|\leq \left\| s_A^{(\varepsilon)}-q_A^{(\varepsilon)}\right\|+\left\|q_A^{(\varepsilon)}-q_A\right\|<\kappa+\frac\varepsilon 2 < \varepsilon.
\]
Thus, we may set $q_{AM}^{\delta,\varepsilon}:=s_{AM}^{(\varepsilon)}$.
\end{proof}

For work extraction, we need a notion of entropy sink. A ``max entropy sink'' $S$ consists of a large collection of states of the form
\[
   s^{(m,n)}:=\left(\underbrace{\frac 1 m,\frac 1 m,\ldots, \frac 1 m}_m,0,0,\ldots,0\right)\in\R^n,
\]
where $m,n\in\N$ and $m\leq n$. We will ``dump max entropy'' into $S$ by transforming these states into
\[
   s^{(m,n,\varepsilon)}:=\left(\underbrace{\frac{1-\varepsilon} m,\ldots,\frac{1-\varepsilon}m}_m,\frac \varepsilon{n-m},\ldots,\frac\varepsilon{n-m}\right)\in\R^n,
\]
where $0<\varepsilon<1$. We assume that the Hamiltonian of the sink is trivial, $\H_S=0$. Then we have the following entropy balance:
\begin{eqnarray*}
\Delta F:=F\left(s^{(m,n)}\right)-F\left(s^{(m,n,\varepsilon)}\right)&=&\frac 1 \beta\left(H(s^{(m,n,\varepsilon)})-H(s^{(m,n)})\right)= \frac 1 \beta \left(\eta(\varepsilon)+\varepsilon\log\frac{n-m}m\right),\\
\Delta F_0:=F_0\left(s^{(m,n)}\right)-F_0\left(s^{(m,n,\varepsilon)}\right)&=&\frac 1 \beta\left(H_0(s^{(m,n,\varepsilon)})-H_0(s^{(m,n)})\right)=\frac 1 \beta(\log n-\log m),
\end{eqnarray*}
where $\eta(\varepsilon)=-\varepsilon\log\varepsilon-(1-\varepsilon)\log(1-\varepsilon)$. In particular, by choosing $\varepsilon$ small enough, we can make $\Delta F$ as small as we like, while keeping $\Delta F_0$ constant. Note that the states $s^{(m,n)}$ have also been introduced in~\cite{Gour}, under the name ``sharp states''.

\begin{theorem}[Extracting work from the system]
Consider a system $A$ with Hamiltonian $\H_A$ and two distributions $p_A$ and $q_A$ such that $F(p_A)> F(q_A)$. Suppose we would like to extract some energy $\Delta>0$ by transforming $p_A$ approximately into $q_A$. Then, for every $\delta,\varepsilon>0$ and every $m,n\in\N$ with $n/m$ large enough, we can find some $\Delta>F(p_A)-F(q_A)-\delta$ and a thermal operation $\mathcal{T}_{\delta,\varepsilon}$ such that
\begin{equation}
   \mathcal{T}_{\delta,\varepsilon}\left(p_A\otimes (1,0)_W\otimes q_M^{\delta,\varepsilon}\otimes s_S^{(m,n)}\right)=q_{AMS}^{\delta,\varepsilon}\otimes (0,1)_W,
\end{equation}
where $\|q_A^{\delta,\varepsilon}-q_A\|<\varepsilon$, $W$ is a work bit with energy gap $\Delta$, $S$ is a max-entropy sink such that $q_S^{\delta,\varepsilon}=s_S^{(m,n,\varepsilon)}$, and $q_{AMS}^{\delta,\varepsilon}$ is a suitable extension of $q_A^{\delta,\varepsilon}$ and $q_S^{\delta,\varepsilon}$. In particular, the work bit transforms from a pure ground state $(1,0)_W$ to a pure excited state $(0,1)_W$ and does not become correlated with $AMS$, but this comes at the expense of dumping an arbitrarily small amount of entropy into $S$. In more detail, ``$n/m$ large enough'' means that the following inequality must hold:
\[
   \log \frac n m >\max\left\{\log 2, \beta F_0(q_A)-\beta F_0(\gamma_A)+\beta F(p_A)-\beta F(q_A)\right\}.
\]
Both $\mathcal{T}_{\delta,\varepsilon}$ and $q_{ABS}^{\delta,\varepsilon}$ depend on $m$ and $n$, which is however suppressed from the notation.
\end{theorem}
\begin{proof}
The proof is very similar to that of Theorem~\ref{ThePerformingWork}. First, similarly as in the proof of Theorem~\ref{ThePerformingWork}, we will choose some $\Delta$ with $F(p_A)-F(q_A)-\delta<\Delta<F(p_A)-F(q_A)$ such that $e^{-\beta\Delta}$ is rational. Consequently, $\gamma_W$ has only rational entries. Let us suppress the dependence from $\delta$ in the notation in the following. We have
\begin{eqnarray}
S\left(q_A\otimes (0,1)_W\otimes s_S^{(m,n,\varepsilon)}\|\gamma_A\otimes\gamma_W\otimes\mu_S\right)&=& S(q_A\|\gamma_A)+S\left((0,1)_W\|\gamma_W\right)+S\left(s_S^{(m,n,\varepsilon)}\|\mu_S\right)\nonumber\\
&=& \beta F(q_A)-\beta F(\gamma_A)+\beta F((0,1)_W)-\beta F(\gamma_W)+\beta F(s_S^{(m,n,\varepsilon)})-\beta F(\mu_S)\nonumber\\
&<& \beta F(p_A)-\beta F(\gamma_A)-\beta F(\gamma_W)+\beta F(s_S^{(m,n,\varepsilon)})-\beta F(\mu_S)-\beta\Delta\nonumber\\
&=& S(p_A\|\gamma_A)+S((1,0)_W\|\gamma_W)+S\left(s_S^{(m,n)}\|\mu_S\right)-\eta(\varepsilon)-\varepsilon\log\frac{n-m}m\nonumber\\
&<& S\left(p_A\otimes (1,0)_W\otimes s_S^{(m,n)}\|\gamma_A\otimes\gamma_W\otimes\mu_S\right).
\label{eqSDiffAgain}
\end{eqnarray}
Similarly as in the proof of Theorem~\ref{ThePerformingWork}, we will now choose some $\kappa$ with $0<\kappa<\varepsilon$ such that we obtain a distribution $\gamma_A^{(\kappa)}$ with all rational entries and maps $\Phi,\bar\Phi:A\to A$ such that $\Phi(\gamma_A)=\gamma_A^{(\kappa)}$ as well as $\bar \Phi(\gamma_A^{(\kappa)})=\gamma_A$ and $\|\Phi(s)-s\|<\kappa$ and $\|\bar \Phi(s)-s\|<\kappa$ for all probability distributions $s\in A$. We also set $p_A^{(\kappa)}:=\Phi(p_A)$. Our $\kappa$ is chosen small enough such that
\[
   S\left(q_A\otimes (0,1)_W\otimes s_S^{(m,n,\varepsilon)}\|\gamma_A^{(\kappa)}\otimes\gamma_W\otimes \mu_S\right)< S\left(p_A^{(\kappa)}\otimes (1,0)_W\otimes s_S^{(m,n)}\|\gamma_A^{(\kappa)}\otimes\gamma_W\otimes\mu_S\right).
\]
Since $S_0(q_A\|\gamma_A)+\beta\Delta=\beta F_0(q_A)-\beta F_0(\gamma_A)+\beta\Delta<\log\frac n m$, we can choose $\kappa$ also small enough to have $S_0(q_A\|\gamma_A^{(\kappa)})+\beta\Delta<\log \frac n m$, because $S_0$ is continuous in the second entry (though not in the first entry). Using additivity of $S_0$ on tensor products, we obtain
\[
\underbrace{\left(S_0(q_A\|\gamma_A^{(\kappa)}\right)}_{<\log\frac n m -\beta\Delta}+\underbrace{S_0((0,1)_W\|\gamma_W)}_{\log Z_W+\beta\Delta}+\underbrace{S_0\left(s_S^{(m,n,\varepsilon)}\|\mu_S\right)}_0<
 \underbrace{S_0\left(p_A^{(\kappa)}\|\gamma_A^{(\kappa)}\right)}_0+\underbrace{S_0((1,0)_W\|\gamma_W)}_{\log Z_W} + \underbrace{S_0\left(s_S^{(m,n)}\|\mu_S\right)}_{\log\frac n m},
\]
which gives us the analog of~(\ref{eqSDiffAgain}) for $S_0$ due to its additivity on tensor products. Similarly as in the proof of Theorem~\ref{ThePerformingWork}, since $\gamma_A^{(\kappa)}\otimes\gamma_W$ has all rational entries, we obtain a corresponding embedding map $\Gamma_\kappa:AW\to(AW)'$. For $\alpha\in\{0,1\}$, it satisfies
\begin{eqnarray*}
H_\alpha\left(\Gamma_\kappa\left(p_A^{(\kappa)}\otimes (1,0)_W\right)\otimes s_S^{(m,n)}\right)	&=& \log D_\kappa -S_\alpha\left(p_A^{(\kappa)}\otimes (1,0)_W\|\gamma_A^{(\kappa)}\otimes \gamma_W\right)+\log n -S_\alpha\left(s_S^{(m,n)}\|\mu_S\right)\\
&=&\log D_\kappa+\log n -S_\alpha\left(p_A^{(\kappa)}\otimes (1,0)_W\otimes s_S^{(m,n)}\|\gamma_A^{(\kappa)}\otimes\gamma_W\otimes\mu_S\right)\\
&<& \log D_\kappa+\log n -S_\alpha\left(q_A\otimes (0,1)_W\otimes s_S^{(m,n,\varepsilon)}\|\gamma_A^{(\kappa)}\otimes\gamma_W\otimes \mu_S\right)\\
&=&H_\alpha\left(\Gamma_\kappa\left(q_A\otimes (0,1)_W\right)\otimes s_S^{(m,n,\varepsilon)}\right).
\end{eqnarray*}
Denoting the image of $AW$ under $\Gamma_\kappa$ by $(AW)'$, we can again invoke Theorem~\ref{TheMaj}, obtaining a distribution $r_{(AW)'MS}^{(\varepsilon)}$ with $\Gamma_\kappa\left(p_A^{(\kappa)}\otimes (1,0)_W\right)\otimes r_M^{(\varepsilon)}\otimes s_S^{(m,n)} \succ r_{(AW)'MS}^{(\varepsilon)}$ and $r_{(AW)'S}^{(\varepsilon)}=\Gamma_\kappa\left(q_A\otimes (0,1)_W\right)\otimes s_S^{(m,n,\varepsilon)}$. Thus, there exists a bistochastic map $\Lambda_\varepsilon:(AW)'MS\to(AW)'MS$ such that
\[
    \Lambda_\varepsilon\left(\Gamma_\kappa\left(p_A^{(\kappa)}\otimes (1,0)_W\right)\otimes r_M^{(\varepsilon)}\otimes s_S^{(m,n)}\right)=r_{(AW)'MS}^{(\varepsilon)}.
\]
Defining a pseudo-inverse $\bar\Gamma_\kappa$ exactly as in the proof of Theorem~\ref{ThePerformingWork}, we can define our linear map this time as
\[
   \mathcal{T}_\varepsilon:=\left(\bar\Phi_A\otimes \id_{WMS}\right)\circ \left(\bar\Gamma_\kappa\otimes\id_{MS}\right)\circ\Lambda_\varepsilon\circ \left(\Gamma_\kappa\otimes\id_{MS}\right)\circ\left(\Phi_A\otimes\id_{WMS}\right).
\]
It is easy to see that $\mathcal{T}_\varepsilon$ maps the thermal state $\gamma_A\otimes\gamma_W\otimes\mu_M\otimes\mu_S$ of $AWMS$ onto itself, hence it is a thermal operation. Furthermore,
\[
   \mathcal{T}_\varepsilon\left(p_A\otimes (1,0)_W\otimes r_M^{(\varepsilon)}\otimes s_S^{(m,n)}\right)=\left(\bar\Phi_A\otimes\id_{WMS}\right)\circ\left(\bar\Gamma_\kappa\otimes\id_{MS}\right)\left(r_{(AW)'MS}^{(\varepsilon)}\right)=:s_{AWMS}^{(\varepsilon)}.
\]
It follows that $s_{MS}^{(\varepsilon)}=r_{MS}^{(\varepsilon)}$, and
\[
   s_{AW}^{(\varepsilon)}=\left(\bar\Phi_A\otimes\id_W\right)\left(\bar\Gamma_\kappa\left(r_{(AW)'}^{(\varepsilon)}\right)\right)=
   \left(\bar\Phi_A\otimes\id_W\right)\left(\bar\Gamma_\kappa\left(\Gamma_\kappa(q_A\otimes (0,1)_W)\right)\right)=\bar\Phi_A(q_A)\otimes (0,1)_W.
\]
Since pure states are uncorrelated with other systems, we get $s_{AWMS}^{(\varepsilon)}=(0,1)_W\otimes s_{AMS}^{(\varepsilon)}$. We also get
\[
   \|s_A^{(\varepsilon)}-q_A\|=\|\bar\Phi_A(q_A)-q_A\|<\kappa<\varepsilon.
\]
Thus, we can set $q_{AMS}^{\delta,\varepsilon}:=s_{AMS}^{(\varepsilon)}$.
\end{proof}

\subsection{Work cost example from Subsection~\ref{SubsecExample}}
\label{AppendixQubit}
Our goal is to determine under what conditions the transition
\begin{equation}
   \gamma_A\otimes\sigma_M\otimes |e\rangle\langle e|_W \longrightarrow \rho'_{AM} \otimes |g\rangle\langle g|_W
   \label{eqlor}
\end{equation}
can be accomplished by a thermal operation, without additional catalyst. Labelling (and sorting) the eigenvectors of $AMB$ by
\[
   |g_A 0 g\rangle, \enspace |g_A 0 e\rangle, \enspace |g_A 1 g\rangle,\enspace |g_A 1 e\rangle,\enspace |e_A 0 g\rangle, \enspace |e_A 0 e\rangle,\enspace |e_A1g\rangle,\enspace |e_A 1 e\rangle,
\]
the state on the left-hand side corresponds to the probability distribution
\[
   p_{AMW}=\left( 0,\frac 1 5,0,\frac 7 {15}, 0 ,\frac 1 {10},0,\frac 7 {30}\right),
\]
and the state on the right-hand side to
\[
   q_{AMW}=\left(\frac 1 {10}, 0, \frac 2 5,0,\frac 1 5,0,\frac 3 {10},0\right).
\]
The sorted energy eigenvalues are
\[
   (E_1,\ldots,E_8)=(0,\Delta,0,\Delta,E_A,E_A+\Delta,E_A,E_A+\Delta),
\]
where $E_A=k_B T \log 2$. We use the \emph{thermomajorization criterion} as explained, for example, in the Supplementary Note E of~\cite{HorodeckiOppenheim}: there exists a thermal operation mapping $p$ to $q$ if and only if the thermal Lorenz curve of $p$ is everywhere on or above the thermal Lorenz curve of $q$. Using Mathematica, we have generated the plots in Figure~\ref{fig_lorenz} for $\Delta=.26 k_B T$, which shows that $p$'s curve (in blue) is indeed nowhere below $q$'s curve (in orange); the same must then be true for larger values of $\Delta$ (and we have numerically verified this). We have also used Mathematica to verify directly the necessary inequalities for all ``elbow points'' of the curves.
\begin{figure}
\begin{center}
\includegraphics[width=.4\textwidth]{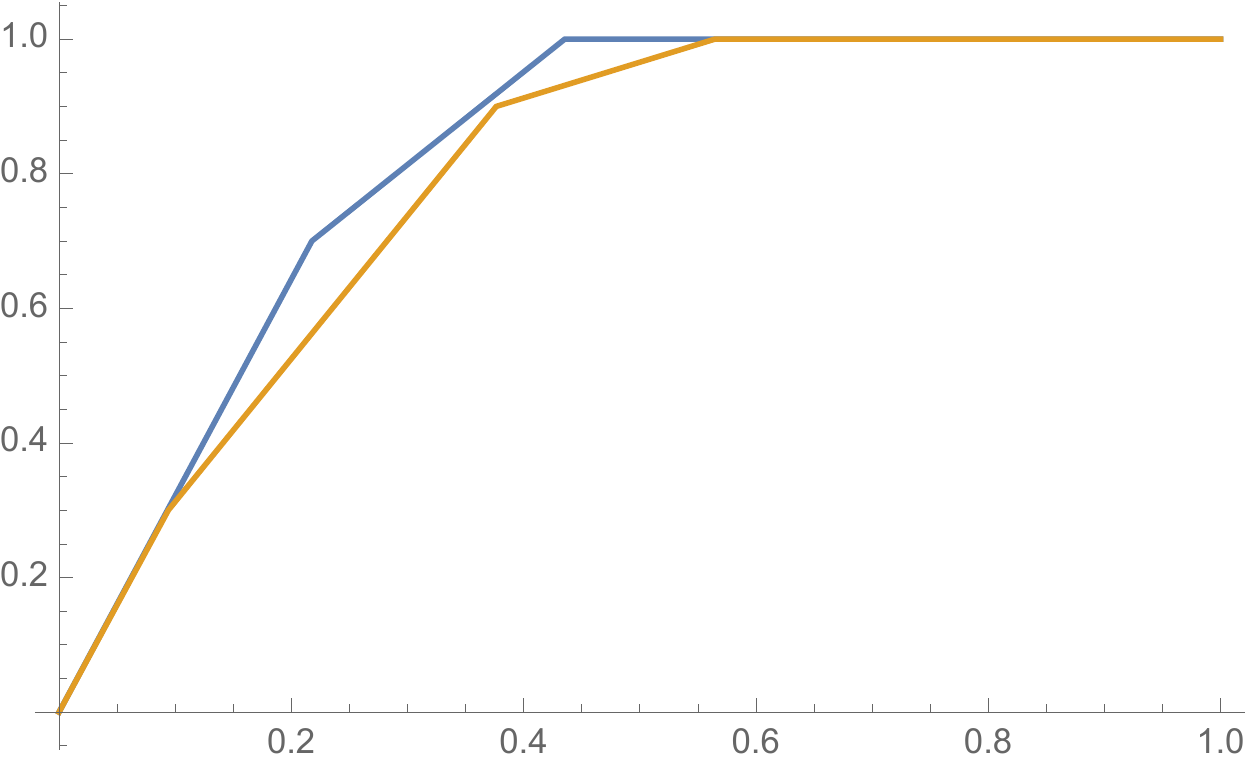}
\end{center}
\caption{\small The thermal Lorenz curves signify the possibility of state transition~(\ref{eqlor}) by a thermal operation.}\label{fig_lorenz}
\end{figure}

\end{document}